\documentclass[smallextended]{svjour3}       
\smartqed  
\pdfoutput=1
\usepackage{graphicx}

\usepackage{amssymb}
\usepackage{eurosym}
\usepackage{amsfonts}
\usepackage{amsmath}
\usepackage{graphicx}
\usepackage{color}
\usepackage{soul}
\usepackage{ulem}
\usepackage[mathscr]{euscript}
\usepackage{xcolor,cancel}
\usepackage{calc}
\usepackage{float}
\usepackage{url}

\usepackage{enumitem}

\usepackage{amsthm}
 \usepackage{mathptmx}      
\usepackage{latexsym}

\newtheorem*{theorem*}{Theorem}

 \journalname{}

\begin{document}

\title{Modeling competitive interactions and plant-soil feedback in vegetation dynamics
}


\author{A. Marasco         \and
        F. Giannino \and A. Iuorio 
}


\institute{A. Marasco \at
              Department of Mathematics and Applications, University of Naples Federico II,  Complesso Universitario di  Monte S. Angelo, via Cintia, Naples,  80126, Italy \\
              \email{marasco@unina.it}           
           \and
           F. Giannino \at
              Department of Agricultural Sciences, University of Naples Federico II, via Universit\`a 100, Portici (NA),  80055, Italy\\
              \email{giannino@unina.it} 
           \and
           A. Iuorio \at
              Johann Radon Institute for Computational and Applied Mathematics (RICAM),
              Vordere Zollamtstra{\ss}e 3, 1030 Vienna, Austria\\
              \email{annalisa.iuorio@ricam.oeaw.tuwien.ac.at}
}

\date{}

\maketitle

\begin{abstract}
Plant-soil feedback is recognized as a causal mechanism for the emergence of vegetation patterns of the same species especially when water is not a limiting resource (e.g. humid environments) \cite{Carteni_2012,Marasco_2014}. Nevertheless, in the field, plants rarely grow in monoculture  but compete with other plant species. In these cases, plant-soil feedback was shown to play a key role in plant-species coexistence \cite{Mazzoleni_2010}.
Using a mathematical model consisting of four PDEs, we investigate mechanisms of inter- and intra-specific plant-soil feedback on the coexistence of two competing plant species. In particular, the model takes into account both negative and positive feedback influencing the growth of the same and the other plant species.
Both the coexistence of the plant species and the dominance of a particular plant species is examined with respect to all model parameters together with the emergence of spatial vegetation patterns. 
\end{abstract}

\section{Introduction} \label{intro}
Understanding vegetation dynamics is fundamental in order to foresee the evolution of an ecosystem, in particular its resilience and ability to preserve biodiversity. To this aim, extensive studies have been performed in the last decades to investigate how plant species interact with each other (through competition for resources) and with the surrounding environment (through positive and negative feedback induced by different environmental factors). In water-limited environments (e.g.~savannas), several ecological mechanisms have been proposed to justify coexistence of herbaceous and woody plant species, such as resource niche separation and ecosystem disturbances (see, for instance, \cite{Baudena_2010,Walter_1971}). Moreover, water-vegetation feedbacks in the form of short-range facilitation and long-range inhibition, as well as self-organisation principles, have been shown to be play a key role in the formation of mono- (see \cite{Klausmeier_1999,Meron_2012,Rietkerk_2002,Sherratt_2005,von_Hardenberg_2001}) and multi-species   (e.g., \cite{Baudena_2012,Eigentler_2019_2,Eigentler_2019,Gilad_2007}) vegetation patterns in arid environments. The above mentioned factors, however, fail to describe the connection between species diversity and latitude - a topic which has been of a focus of interest of the scientific community for a long time \cite{Darwin_1909,Weigelt_2003,Willig_2003}. \\
A factor that proved in recent years to play an important role in shaping natural plants communities is plant-soil negative feedback (see, for instance,~\cite{Bever_1994,Eppinga_2018,van_der_Putten_1993}). Causes for this negative feedback include the presence of soil pathogens, the changing composition of soil microbial communities~\cite{Bever_1994,Kulmatiski_2008,van_der_Putten_1993} and the accumulation of autotoxic compounds from decomposing plant litter~\cite{Bonanomi_2011,Mazzoleni_2007}. Plant-soil negative feedback has been shown to be very relevant in spatio-temporal dynamics of plant systems, particularly in the spatial organisation of plants by means of clonal rings~\cite{Bonanomi_2014,Carteni_2012} and patterns~\cite{Marasco_2014}, and in species coexistence \cite{Bonanomi_2005,Eppinga_2018,Mazzoleni_2010}. The governing process in both cases is the attempt of the biomass to ``escape'' areas with a high concentration of toxic compounds. For species growing in humid environments (e.g.~clonal plants), this area is located at the center of the tussock \cite{Bonanomi_2014}. These studies are able to reveal interesting features: in environments where litter decomposition is rapid and nutrient cycles are closed we have species' richness, on the other hand we have single species dominance where litter decomposition is slow and/or negative feedback is removed from the nutrient cycle pathway. Nevertheless, so far they have only focused on the effect of intra-specific plant-soil negative feedback induced by toxicity. Moreover, a combination of inter-specific feedback induced by toxicity and spatial effects -- which are known to play a prominent role in self-organisation processes for vegetation (see \cite{Eigentler_2019,Kealy_2011,Marasco_2014} and references therein) -- has not been considered before.

In this work, we consider for the first time both inter- and intra-specific plant-soil feedbacks induced by toxicity in a spatial model describing the interaction of two plant species, and their two relative \textit{toxicities}. We assume that the nature of the inter-specific feedback can be positive or negative, depending on the phylogenetic distance between the two species \cite{Mazzoleni_2014}. We will show how the interplay of these effects can lead to the emergence of alternative ecological scenarios  such as competitive exclusion and species coexistence.  This last scenario can occur both when the  species live in the same area at any time, as well as when they are continuously alternating in space and time. In mathematical terms, this corresponds to the emergence of stable coexistence equilibria and spatio-temporal patterns, respectively.

The paper is structured as follows: the mathematical model and the main ecological assumptions are introduced in Section \ref{sec:2}, while its nondimensional version is presented in Section \ref{sec:3}. Section \ref{sec:4} is devoted to the investigation of ecologically feasible equilibria and linear stability analysis to spatially homogeneous perturbations, whereas spatially heterogeneous perturbations are considered in Section \ref{sec:5}. Numerical simulations  concerning the influence of the main parameters on dynamic of the model are illustrated in Section \ref{sec:6}. Conclusions and research perspectives conclude the manuscript.

\section{The mathematical model} \label{sec:2}
To investigate the effects of inter- and intra-specific plant-soil feedback on the
coexistence of two competing plant species in an environment where water is not a limited resource,  we consider the biomasses $B_1$ and $B_2$ and the corresponding toxicities $T_1$ and $T_2$ in a bidimensional domain (see Fig.~\ref{fig:moddes}). 
\begin{figure}[H]
    \centering
    \includegraphics[scale=0.5]{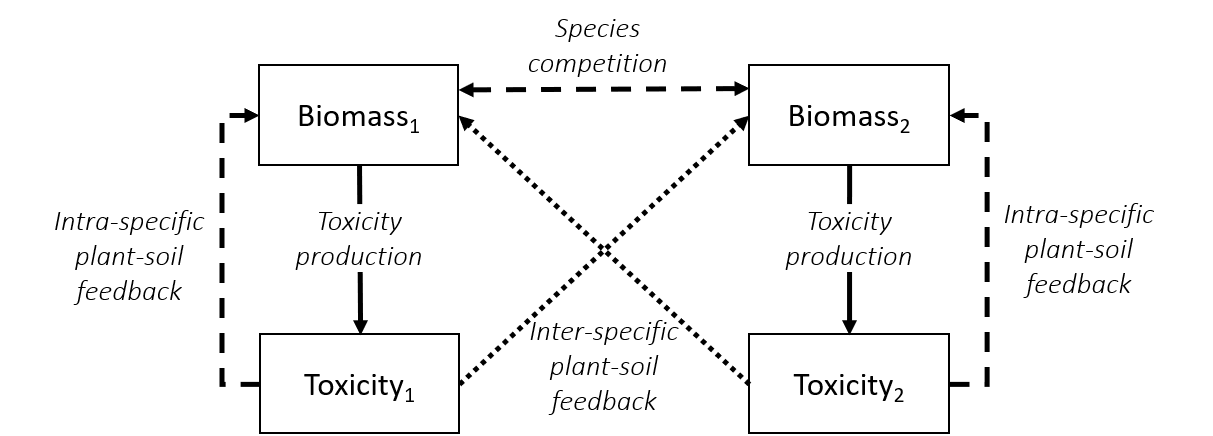}
    \caption{Schematic representation of the examined ecological scenario. Solid lines represent (positive) mass flows, dashed lines indicate the (negative) effect exerted by auto-toxicity on the species producing it, and dotted lines describe the (positive or negative) interaction of the biomass with toxicity produced by the other species.}
    \label{fig:moddes}
\end{figure}
We assume that the plant biomass of the $i$-th species $B_i$ evolves according to a growth rate parameter $a_i$ (time$^{-1}$), a constant death rate $d_i$ (time$^{-1}$), and an extra loss due to the negative plant–soil feedback induced by $T_i$ by means of $s_{ii}$ (m$^2$ kg$^{-1}$ time$^{-1}$). An interspecific competition between $B_i$ and $B_j$ -- whose strength is measured by the parameters $a_{ij}$ -- is also included. Moreover, we assume that the toxicity
produced by each species can have either a positive or a negative feedback on the other species, and that the situation is symmetrical, i.e., if $T_1$ influences $B_2$ positively, so does $T_2$ with $B_1$, and vice versa. To this aim, we introduce the coefficient $s_{ij}$ representing the sensitivity of species $i$ to the toxicity produced by the species $j$. Moreover, assuming that $sign(s_{ji} )= sign(s_{ij})$ for all $i\neq j$, we have
\begin{enumerate}
\item if $s_{ij}<0$, the toxicity $T_j$ has a positive feedback on the species $B_i$, and consequently $s_{ji} < 0$. For instance, $B_i$ can find additional nutrients in the toxicity produced by $B_j$;

\item if $s_{ij} > 0$, the toxicity $T_j$ has a negative feedback on the species $B_i$, and also $s_{ji}>0$;

\item if $s_{ij}=0$, the toxicity $T_j$ do not affect the dynamics of $B_i$, and consequently $s_{ji} = 0$. 
\end{enumerate}
Finally, plants vegetative spatial dispersal is modelled by a diffusion term with coefficient $D_i$ (time$^{-1}$). \\
The concentration of toxic compounds $T_i$ produced by species $i$ is determined by a fraction $c_i$ (time$^{-1}$) of the biomass $B_i$ and decreases because of litter removal/decay processes, which are represented by the parameter $k_i$ (time$^{-1}$). For sake of simplicity, no diffusion of toxicity is considered by the model. 

Then, the PDEs model write
\begin{equation}
\left\{
\begin{array}{l}
\displaystyle\frac{\partial B_{1}}{\partial t}=\underset{\text{growth}}{\underbrace{%
a_{1}B_{1}}}-\underset{\text{mortality and extra-mortality}}{%
\underbrace{\left( d_{1}+s_{11}T_{1}\right) B_{1}}}-\underset{\text{%
competition with }B_{2}}{\underbrace{a_{12}B_{1}B_{2}}}-\underset{\text{interaction with }T_{2}}{\underbrace{s_{12}T_{2}B_{1}}}+\underset{\text{%
dispersal}}{\underbrace{D_{1}\Delta B_{1}}},\medskip  \\
\displaystyle\frac{\partial B_{2}}{\partial t}=\underset{\text{growth}}{\underbrace{%
a_{2}B_{2}}}-\underset{\text{mortality and extra-mortality}}{%
\underbrace{\left( d_{2}+s_{22}T_{2}\right) B_{2}}}-\underset{\text{%
competition with }B_{1}}{\underbrace{a_{21}B_{1}B_{2}}}-\underset{\text{interaction with }T_{1}}{\underbrace{s_{21}T_{1}B_{2}}}+\underset{\text{%
dispersal}}{\underbrace{D_{2}\Delta B_{2}}},\medskip  \\
\displaystyle\frac{\partial T_{1}}{\partial t}=\underset{\text{production}}{\underbrace{%
c_{1}B_{1}}} -\underset{\text{decomposition}}{\underbrace{%
k_{1}T_{1}},}\medskip
\\
\displaystyle\frac{\partial T_{2}}{\partial t}=\underset{\text{production}}{\underbrace{%
c_{2}B_{2}}} -\underset{\text{decomposition}}{\underbrace{%
k_{2}T_{2}},}%
\end{array}%
\right. \label{1}
\end{equation}
where all coefficients are positive except for $s_{12}$ and $s_{21}$. 

For a detailed description of the model variables and parameters, together with their units, we refer to Table~\ref{tab:par}.

\begin{table}[ht!]
\centering
\begin{tabular}{ccc}
\hline\hline \textbf{Symbol} & \textbf{Description} & \textbf{Unit}
\\ \hline\hline $B_{i}$ & \multicolumn{1}{l}{plant biomass of $i$-th
species} &
\multicolumn{1}{l}{kg m$^{-2}$} \\
$T_{i}$ & \multicolumn{1}{l}{toxicity produced by $i$-th species} &
\multicolumn{1}{l}{kg m$^{-2}$} \\
$a_{i}$ & \multicolumn{1}{l}{growth rate of $B_{i}$} & \multicolumn{1}{l}{
time$^{-1}$} \\
$d_{i}$ & \multicolumn{1}{l}{death rate of $B_{i}$} & \multicolumn{1}{l}{time%
$^{-1}$} \\
$s_{ij}$ & \multicolumn{1}{l}{sensitivity of $B_{i}$ to $T_{j}$} &
\multicolumn{1}{l}{m$^{2}$ kg$^{-1}$ time$^{-1}$} \\
$a_{ij}$ & \multicolumn{1}{l}{competition between $B_{i}$ and $B_{j}$} &
\multicolumn{1}{l}{m$^{2}$ kg$^{-1}$ time$^{-1}$} \\
$D_{i}$ & \multicolumn{1}{l}{dispersal coefficient of $B_{i}$} &
\multicolumn{1}{l}{time$^{-1}$} \\
$c_{i}$ & \multicolumn{1}{l}{proportion of toxic products by litter
decomposition of $i$-th species} & \multicolumn{1}{l}{time$^{-1}$} \\
$k_{i}$ & \multicolumn{1}{l}{decay rate of $T_{i}$} & \multicolumn{1}{l}{time%
$^{-1}$} \\ \hline
\end{tabular}
\caption{List of model variables and parameters with their units}
\label{tab:par}
\end{table}

System~(\ref{1}) is defined on the bounded domain $\Omega\subset \mathbb{R}^{2}$ and is equipped with the following initial and boundary conditions
\begin{subequations} \label{eq:ombic}
\begin{gather}
\begin{split}
 B_{i}\left( \mathbf{x},0\right) =B_{i,0}\left( \mathbf{x}\right)
,\quad T_{i}\left( \mathbf{x},0\right) =T_{i,0}\left( \mathbf{x}%
\right), \quad \quad \quad \mathbf{x}\in \Omega, \label{IC}
 \\
\end{split}
\intertext{and}
 \partial _{n}B_{i}=0,\quad \partial _{n}T_{i}=0,\quad \quad \quad \mathbf{x}%
\in \partial \Omega ,\quad t\in \mathbb{R}^{+}, \label{BC}
\end{gather}
\end{subequations}
where $\partial \Omega$ is the boundary of $\Omega$, $\partial _{n}$ is the normal derivative on $\partial \Omega$, and $i=1,2$.

\subsection{Assumptions on the model parameters} \label{ssec:ass}
From now on, we focus our attention on an ecological scenario in which two \textit{similar species} (i.e., phylogenetically close and with the same capacity to interact with the surrounding environment) interact in the same habitat (e.g. two shrubs species). In this context, we assume that:
\begin{itemize}
\item The intrinsic growth rates of $B_{1}$ and $B_{2}$ are equal, i.e.,
\begin{equation}
   a_{1}-d_1=a_{2}-d_2. \label{A1}
\end{equation}

\item The decay rates of $T_{1}$ and $T_{2}$ are equal, i.e.,
\begin{equation}
   k_{1}=k_{2}. \label{A2}
\end{equation}
\item The sensitivity of $B_{i}$ to $T_{j}$ is proportional to the sensitivity of $B_{j}$ to $T_{j}$, for all $i$ and $j\neq i$, i.e.,
\begin{equation}
   s_{12}=K s_{22}, \quad s_{21}=K s_{11}, \label{A3}
\end{equation} where $K$ is a constant (positive, negative or null). This is equivalent to assuming that the effect induced by $T_j$ on $B_i$ (represented by $s_{ij}$) is proportional to the effect that  $T_j$ has on the biomass $B_j$ (i.e. $s_{jj}$). Owing to system~\eqref{1}, when $K>0$ we have that $T_j$ has a negative influence on $B_i$, while if $K<0$ the presence of $T_j$ enhances the growth of $B_i$. On the other hand, if $K=0$ no inter-specific toxicity-biomass interactions are considered.
\end{itemize}

\section{Nondimensional analysis} \label{sec:3}

We remark that system (\ref{1}) depends on $16$ parameters. Assuming conditions (\ref{A1})--(\ref{A3}) the resulting system will depend on $11$ parameters as follows
\begin{equation}
\left\{
\begin{array}{l}
\displaystyle\frac{\partial B_{1}}{\partial t}=
g_{1}B_{1}-s_{11}T_{1}B_{1}-a_{12}B_{1}B_{2}-Ks_{22}T_{2}B_{1}+ D_1 \Delta B_1,\medskip  \\
\displaystyle\frac{\partial B_{2}}{\partial t}=
g_{1}B_{2}- s_{22}T_{2} B_{2}-a_{21}B_{1}B_{2}-Ks_{11}T_{1}B_{2}+ D_2 \Delta B_2,\medskip  \\
\displaystyle\frac{\partial T_{1}}{\partial t}=c_{1}
B_{1}-k_{1}T_{1},\medskip
\\
\displaystyle\frac{\partial T_{2}}{\partial t}=c_{2}
B_{2}-k_{1}T_{2},%
\end{array}%
\right. \label{5}
\end{equation}
where $g_{1}=a_1-d_1$.

We note that the functions $B_i$ and $T_i$, $i=1,2$, evolve on two different time scales. Then, if we perform a nondimensional analysis introducing two reference quantities for the time, we will obtain a dimensionless system depending on $4$ parameters. 

In detail, we introduce the following reference quantities for the biomasses and toxicities
\begin{equation}
\displaystyle B_{0,1}= \frac{g_1}{a_{21}},\quad B_{0,2}= \frac{g_1}{a_{12}},\quad T_{0,1}=\frac{a_1-d_1}{s_{11}},\quad T_{0,2}=\frac{a_1-d_1}{s_{22}},
\label{6}
\end{equation}
and we choose two reference quantities for the time evolution of $B_i$ and $T_i$, respectively, and a reference length as follows
\begin{equation}
\displaystyle \tau_B = \frac{1}{a_1-d_1},\quad\tau_T = \frac{a_{21}}{c_1 s_{11}}, \quad l= \sqrt{\frac{D_1}{a_1-d_1}}.
\label{7}
\end{equation}

From \eqref{7}, we introduce the nondimensional quantities $\hat{t}_B =t/\tau_B$, $\hat{t}_T =t/\tau_T$, and $\hat{Q_i}=Q_i/Q_{0,i}$, as well as the nondimensional parameters
\begin{equation}
 K=\frac{s_{12}}{s_{22}}=\frac{s_{21}}{s_{11}}, \quad T=\frac{a_{21} c_2 s_{22}}{a_{12} c_1 s_{11}}, \quad S=\frac{a_{21} k_1}{c_1 s_{11}}, \quad  D_B=\frac{D_2}{D_1}.
\label{9}
\end{equation}
It is then an easy exercise to verify that the nondimensional form of Eqs.~\eqref{5} is
\begin{equation}
\left\{
\begin{array}{l}
\displaystyle\frac{\partial B_{1}}{\partial t_B} = B_1-B_1 B_2-B_1 T_1-K B_1 T_2 + \Delta B_1,\medskip  \\
\displaystyle\frac{\partial B_{2}}{\partial t_B} = B_2-B_1 B_2 -B_2 T_2 - K B_2 T_1 + D_B \Delta B_2,\medskip  \\
\displaystyle\frac{\partial T_{1}}{\partial t_T} = B_1-S T_1,\medskip
\\
\displaystyle\frac{\partial T_{2}}{\partial t_T} = T B_2-S T_2,%
\end{array}%
\right. \label{8}
\end{equation}
where we have dropped the superscript for sake of simplicity.

All nondimensional parameters in Eq.~(\ref{9}) are positive except for $K$ that can be negative, positive, or null  (see Section \ref{ssec:ass}). We remark that $K$ is the ratio of the effect of $T_j$ on species $B_i$ over the effect of $T_j$ on species $B_j$, i.e., the ratio between the interspecific toxicity and autotoxicity.
Moreover, $T$ is a parameter that measures the impact of competition and autotoxicity on $B_2$ in relation to the impact of the same factors on $B_1$ mediated by the growth rate of each toxicity. In particular, when $T>1$ the above negative effects on species 2 are stronger than on species 1, whereas for $T=1$ such effects are indistinguishable. The parameter $S$ plays the role of the decomposition rate of the toxicities in this nondimensional scheme. Finally,  $D_B$ represents an indirect measure of  the dispersal rates of the biomasses.

The parameter $T$ plays a key role in the dynamics of the model. Consequently, we analyze the existence and the stability properties of coexistence equilibria both for $T=1$ and $T\neq 1$.

\section{Equilibria and  linear stability analysis under spatially homogeneous perturbation}\label{sec:4}
The corresponding nondimensional local system of (\ref{8}) writes
\begin{equation}\label{8ODE}
\left\{
\begin{array}{l}
\displaystyle\frac{d B_{1}}{d t_B} = B_1-B_1 B_2-B_1 T_1-K B_1 T_2, \medskip\\
 \displaystyle\frac{d B_{2}}{d t_B} = B_2-B_1 B_2 -B_2 T_2 - K B_2 T_1, \medskip\\
 \displaystyle\frac{d T_{1}}{d t_T} = B_1-S T_1, \medskip\\
 \displaystyle\frac{d T_{2}}{d t_T} = T B_2-S T_2,
\end{array}
\right.
\end{equation}
where, for convenience, we omitted the superscript.
\subsection{Equilibria}
Biologically feasible homogeneous equilibrium
configuration $E_i\equiv\left(B^{*}_{1,i},B^{*}_{2,i},T^{*}_{1,i},T^{*}_{2,i}\right)$ of Eqs.~(\ref{8ODE}) are non-negative solutions of the following system of algebraic equations
\begin{equation}
\left\{
\begin{array}{l}
{B}_1-{B}_1 {B}_2-{B}_1 {T}_1-K {B}_1 {T}_2 =0, \medskip\\
 {B}_2-{B}_1 {B}_2 -{B}_2 {T}_2 - K {B}_2 {T}_1 =0, \medskip\\
{B}_1-S {T}_1=0, \medskip\\
{T} {B}_2-S {T}_2=0.
\end{array}
\right.\label{10}
\end{equation}
System (\ref{8ODE}) always admits three (non-negative) \textit{non-coexistence equilibrium configurations} 
\begin{equation}\label{11}
E_0= \left(0,0,0,0\right), \quad
E_1=\left(0,\frac{S}{T},0,1\right), \quad
E_2=\left(S,0,1,0\right),
\end{equation}
whereas it can exhibit only one or infinite coexistence equilibria as shown in the following result.

\begin{proposition} [\textit{Existence of the coexistence equilibria}] \label{prop1} Suppose that $S,T>0$. 
\begin{enumerate}[label=(\roman*)]
    \item If $K\geq 1$, system (\ref{8ODE}) admits one and only one coexistence equilibrium 
    \begin{equation}\label{12}
        E_3=\left(S\Psi,S\Theta,\Psi,T\Theta\right),
        \end{equation}
    where \begin{equation}\label{12b}\Psi=\frac{(K-1) T+S}{(K+S) (K T+S)-T}, \quad \Theta=\frac{K+S-1}{(K+S) (K T+S)-T}.\end{equation}
    \item If $K<1$ and $K\neq 1-S$, system (\ref{8ODE}) exhibits the coexistence equilibrium $E_3$ provided that at least one of the following mutually exclusive conditions is satisfied:
    \begin{enumerate}
        \item $-1\leq K<1,\;S<1-K, \; \displaystyle T>\frac{S}{1-K}$ \; \; \; or \; \; \;$-1\leq K<1, \;S>1-K, \; \displaystyle T<\frac{S}{1-K}$,\\
        \item $K < -1,\;\displaystyle S < \frac{1 - K^2}{K},\; \displaystyle \frac{S}{1- K} < 
     T < \displaystyle\frac{S (K + S)}{1-K (K + S)}$,
     \item $K < -1,\;\displaystyle\frac{1-K^2}{K}\leq S<1-K, \;\displaystyle T>\frac{S}{1-K}$ \; \; \; or \; \; \;$K < -1,\;\displaystyle S>1-K,\;T<\frac{S}{1-K}.$
    \end{enumerate}
    
    \item  If $K<1$ and $K=1-S$, system (\ref{8ODE}) has infinite coexistence equilibria 
\begin{equation}\label{13}
  \left(B^*_1,\;S-B^*_1,\; \frac{B^*_1}{S},\;\frac{S-B^*_1}{S}\right), \quad 0<B^*_1<S,
\end{equation}
provided that $T=1$.
\end{enumerate}
 \end{proposition}
\begin{proof}
 The items \textit{(i)} and \textit{(ii)} can be easily deduced noting that the equilibrium $E_3$ is admissible only when  $\Psi$ and $\Theta$ defined in Eq.~(\ref{12b}) are positive, whereas item \textit{(iii)} follows from solving Eq.~(\ref{10}) when $K=1-S$ and $T=1$. 
 \end{proof}

\subsection{Linear stability analysis of  non-coexistence equilibrium configurations}
In this section, we determine the conditions under which the equilibrium configurations (\ref{11}) are locally unstable or asymptotically stable. 

To analyze the linear stability of equilibria we note that the Jacobian matrix of the system (\ref{10}) is given by
\begin{equation} \label{jacG}
J=\left(
\begin{array}{cccc}
 -{B_2}-{T_1}-K {T_2}+1 & -{B_1} & -{B_1} & -{B_1} K \\
 -{B_2} & -{B_1}-K {T_1}-{T_2}+1 & -{B_2} K & -{B_2} \\
 1 & 0 & -S & 0 \\
 0 & T & 0 & -S \\
\end{array}
\right).
\end{equation}

\hspace{1in}

\begin{proposition} [Linear stability properties of $E_0,E_1,E_2$] \label{prop2}
Let $S,T$ be positive and $K$ be a real parameter. Then, we have
\begin{enumerate}
    \item[(1)] The equilibrium $E_0$ (bare soil) is always locally unstable.
    \item[(2)] The equilibrium $E_1$ is locally asymptotically stable if and only if  $\displaystyle K>1-\frac{S}{T}$. 
    
    \item[(3)] The equilibrium $E_2$ is locally asymptotically stable if and only if $K>1-S$.
    
    \item[(4)] system (\ref{8ODE}) undergoes a transcritical bifurcation at $E_1$ when $\displaystyle K\equiv K_{TB,1}=1-\frac{S}{T}$, that is backward if $\displaystyle T < \frac{S}{1+S}$, and forward if $\displaystyle T > \frac{S}{1+S}$.\vspace{0.2cm} Moreover, an additional backward transcritical bifurcation occurs at $E_2$ when $K\equiv K_{TB,2}=1-S$. In both cases, the exchange of stability properties occurs with respect to $E_3$.
\end{enumerate}
 \end{proposition}
 
\begin{proof}
\textit{(1)} The Jacobian matrix evaluated at $E_0$ admits the eigenvalues $\{1,1,-S,-S\}$; therefore, $E_0$ is always unstable.\\
\textit{(2)} The eigenvalues of the Jacobian matrix evaluated at $E_1$ are: $$ \left\{-S,\;1-K-\frac{S}{T},\;-\frac{1}{2} \left(S+\sqrt{(S-4) S}\right),\;-\frac{1}{2} \left(S-\sqrt{(S-4) S}\right)\right\},$$ and it can easily be verified that $E_1$ is asymptotically stable if $K>1-\displaystyle\frac{S}{T}$, whereas it is unstable if $\displaystyle K<1-\frac{S}{T}$.\\
\textit{(3)} The Jacobian matrix at $E_2$ has the following eigenvalues $$\left\{-S,\;1-K-S,\;-\frac{1}{2} \left(S+\sqrt{(S-4) S}\right),\;-\frac{1}{2} \left(S-\sqrt{(S-4) S}\right)\right\}.$$ Therefore $E_2$ is asymptotically stable if $K>1-S$, and is unstable if $K<1-S$.\\
\textit{(4)}
To determine the stability properties of the nonhyperbolic equilibria $E_1$ and $E_2$ for $K=1-\displaystyle\frac{S}{T}$ and $K=1-S$, respectively, and the emergence of another equilibrium (bifurcated from each of them) we resort to Theorem~4.1 of \cite{Chavez_2004} based on the center manifold theory. In particular, we investigate the signs of  the coefficients $a$ and $b$ of the normal form of the system on the center manifold (see Appendix~A).

For $\displaystyle K=1-\frac{S}{T}$, the Jacobian matrix \eqref{jacG} evaluated at $E_1$ admits the eigenvalues
\begin{equation}
    0,\;-S,\;-\frac{1}{2} \left(S \pm \sqrt{(S-4) S}\right),
\end{equation}
which are all negative except for a single zero eigenvalue, thus the center manifold is one-dimensional. The right and left eigenvectors corresponding to the zero eigenvalue, namely $\mathbf{w}$ and $\mathbf{v}$, are
    \begin{equation}
    \begin{aligned}
        \mathbf{w} &= \left( \frac{S T}{S - T - S T},~\frac{S}{T},~\frac{T}{S-T-S T},~1 \right),\\
        \mathbf{v} &= \left(\frac{S-T-S T}{ST},~0,~0,~0 \right).
    \end{aligned}
    \end{equation}
    Correspondingly, the quantities $a$ and $b$ read
    \begin{equation}
    \begin{aligned}
        a &= \frac{S^4(T-1)^2+(S T)^2(T-2)^2+2 T^4+S T^3(2T-3)-S^3 T(3-4T+T^2)}{S T^3(T-S+S T)}, \\
        b &= -1,
    \end{aligned}
    \end{equation}
    where $a<0$ if $\displaystyle T<\frac{S}{1+S}$ (which can only occur for $T<1$), while $a>0$ for $\displaystyle T>\frac{S}{1+S}$.\\
    Then, using \cite[Theorem~4.1]{Chavez_2004} we have
    \begin{itemize}
        \item Let $\displaystyle T < \frac{S}{1+S}$. If $\displaystyle K <  K_{TB,1}$ and $\displaystyle \left|K-K_{TB,1} \right| \ll 1$, then $E_1$ is unstable. Moreover, if $\displaystyle 0 <K-K_{TB,1} \ll 1$, then $E_1$ is locally asymptotically stable, and there exists a negative (unfeasible) unstable equilibrium, corresponding to $E_3$. 
       
        \item Let $\displaystyle T > \frac{S}{1+S}$. If $\displaystyle K < K_{TB,1}$ and $\displaystyle \left|K-K_{TB,1} \right| \ll 1$, $E_1$ is unstable  and there exists a locally asymptotically stable negative (unfeasible) equilibrium. In addition, if $\displaystyle 0 < K-K_{TB,1} \ll 1$, $E_1$ is locally asymptotically stable and the positive unstable equilibrium $E_3$ appears.
    \end{itemize}

If $K=1-S$, the eigenvalues of the Jacobian matrix \eqref{jacG} evaluated at $E_2$ are
\begin{equation}
    0, \;-\frac{1}{2} \left(S \pm \sqrt{(S-4) S}\right),\; -S,
\end{equation}
which are all negative, except for a simple zero eigenvalue. The right and left eigenvectors corresponding to the zero eigenvalue, namely $\mathbf{w}$ and $\mathbf{v}$, write
    \begin{equation}
    \begin{aligned}
        \mathbf{w} &= \left( \frac{S~(S T - T - S)}{T},~\frac{S}{T},~\frac{S T - T - S}{T},~1 \right),\\
        \mathbf{v} &= \left( 0,~\frac{T}{S},~0,~0 \right).
    \end{aligned}
    \end{equation}
    Consequently, the quantities $a$ and $b$ read
    \begin{equation}
    \begin{aligned}
        a &= -\frac{T^2-(S-1-S^2)(S+T-S T)^2}{S T}, \\
        b &= -1,
    \end{aligned}
    \end{equation}
    which are both negative for every positive $S$ and $T$. Then, if $K<K_{TB,2}$ and $|K-K_{TB,2}| \ll 1$, $E_2$ is unstable, whereas if $0 < K-K_{TB,2} \ll 1$, $E_2$ is locally asymptotically stable  and there exists the unstable equilibrium $E_3$ provided that $0<B^{*}_{1,3}<B^{*}_{1,2}$, i.e. if $K>K_{TB,1}$ (see Remark~\ref{rem:pos}). 
\end{proof}

\begin{remark} \label{rem:pos}
We note that Theorem 4.1 in \cite{Chavez_2004} refers to the branch of zero equilibria that exchange the stability properties with an emerged second branch of equilibria via a transcritical bifurcation. Hence, in this context \textquotedblleft negative" and \textquotedblleft positive" should be interpreted as \textquotedblleft less than" or \textquotedblleft greater than" the corresponding component on the zero-branch. In our case, since $B^*_{1,1}=0$, the emergence of a branch of negative equilibria leads to an ecologically unfeasible scenario. On the contrary, being $B^*_{1,2}=S$, the emergence of a new branch of equilibria $E_3$  occurs provided that $B^*_{1,3}<S$. This condition is always satisfied for $\displaystyle K>K_{TB,2}$ when  $E_3$ exists (see Proposition~\ref{prop1} \textit{(ii)}). Finally, we note that the analytical results of Proposition~\ref{prop2} are in agreement with the (numerical) bifurcation diagrams reported in Figure~\ref{fig:bifdiag} (for $S=0.5$, and $T=0.1,1,1.5$).
\end{remark}

\begin{remark}
The result \textit{(4)} in Proposition~\ref{prop2} gives us indirect information on the local stability properties of the coexistence equilibrium $E_3$ in a neighbourhood of each transcritical bifurcation point $K_{TB,1}$ and $K_{TB,2}$, for every $T$. The local stability properties of $E_3$ when $T=1$ will be discussed in more details in the following section, where we will also analytically prove the existence of a Hopf bifurcation.
\end{remark}

We conclude this section noting that the equilibrium value $B_{2,1}^*$, whenever $E_1$ is asymptotically stable, increases with respect to $S$ and decreases with respect to $T$. Analogously, the equilibrium value $B_{1,2}^*$ of $E_2$ always increases with respect to $S$.

\subsection{Linear stability analysis of coexistence equilibria in the case $T=1$} \label{ssec:4.3}

When $T=1$ the transcritical bifurcation values $K_{TB,1}$ and $K_{TB,2}$ introduced in Proposition \ref{prop2} coincide. However, the stability analysis of the coexistence equilibrium $E_3$ shows the occurrence of an additional \textit{simple Hopf bifurcation}\footnote{We refer to a simple Hopf bifurcation when a pair of complex conjugate eigenvalues of the Jacobian matrix passes through the imaginary axis while all other eigenvalues have negative real parts.}.

\begin{proposition} \label{prop3}
Let $S>0$ and $T=1$, then the coexistence equilibrium (\ref{12}) reduces to
    \begin{equation}\label{15}
       \bar{E}_3=\left( \frac{S}{K+S+1},\frac{S}{K+S+1},\frac{1}{K+S+1},\frac{1}{K+S+1}\right). 
    \end{equation}

\begin{enumerate}
    \item[(1)] If $K\geq 1$ the equilibrium $\bar{E}_3$ always exists and is locally unstable.
    \item[(2)] If $K<1$ and $K\neq 1-S$ the equilibrium (\ref{15}) exists if  
    \begin{equation}\label{17}
     -(1+S)<K<1-S \quad \textit{or} \quad 1-S<K<1,   
    \end{equation}  and  
    \begin{enumerate}
        \item [(2a)] $\bar{E}_3$ is locally unstable if \; $-(1+S)<K<-S$ or $1-S<K<1$;
        \item [(2b)] $\bar{E}_3$ is locally asymptotically stable if \;$-S<K<1-S$;
        \item [(2c)] system (\ref{8ODE})  undergoes a \textit{simple Hopf bifurcation} at  $\bar{E}_3$ when $K \equiv \bar{K}_{HB}=-S$;
        \item [(2d)] When $K \equiv \bar{K}_{TB}=1-S$, a transcritical bifurcation occurs, and the non-coexistence equilibria exchange the stability properties with $\bar{E}_3$.
    \end{enumerate}
   \item[(3)] If $K=\bar{K}_{TB}$ the coexistence equilibria (\ref{13}) are always locally unstable. 
\end{enumerate}
\end{proposition} 
\begin{proof}
 \textit{(1)} \;The eigenvalues of the Jacobian matrix evaluated at $\bar{E}_3$ are:
 \begin{equation} \label{18}
 \begin{aligned}
    &-\frac{S (K+S+2)\pm \sqrt{S \left(S (K+S)^2-4 (K+1) (K+S+1)\right)}}{2 (K+S+1)},\\
    &-\frac{S (K+S)\pm \sqrt{S \left((S+4) (K+S)^2-4\right)}}{2 (K+S+1)}
 \end{aligned}
 \end{equation}
 and are always positive if $K\geq 1$.
 
 \vspace{0.5cm}
 \textit{(2)} \;The conditions (\ref{17}) can easily be deduced from the conditions \textit{(iia)} of Proposition \ref{prop1} by recalling that $T = 1$.  Let us denote the characteristic polynomial of the Jacobian matrix evaluated at $\bar{E}_3$ by
 
 \begin{equation} \label{19}
  \lambda^4 +A_1  \lambda^3 +A_2  \lambda^2+A_3 \lambda +A_4.
 \end{equation}
  According to Routh--Hurwitz criterion, the polynomial (\ref{19}) has all roots with negative real parts if and only if the following relations hold
 \begin{equation} \label{20}
 A_1,\,A_3,\,A_4>0, \quad A_1 A_2-A_3>0, \quad A_1 A_2 A_3-{A_3}^{2}-{A_1}^{2} A_4>0.   
 \end{equation}
 In our case, we have
 \begin{equation} \label{21}
  A_1=2S, \;A_2=S \left[\frac{2 K+S+2}{(K+S+1)^2}+S\right], \; A_3=\frac{2 S^2}{(K+S+1)^2}, \; A_4=\frac{S^2 (1-K-S)}{1+K+S},
 \end{equation}
 and conditions (\ref{20}) reduce to the following 
\begin{equation}
    S>0, \quad (K+S)^2<1, \quad (K+S) (K+S+2) \left[(K+S)^2+S\right]>0.
\end{equation}
Thus, \textit{(2a)} and \textit{(2b)} are proved. 

Similarly, we will prove that system (\ref{10}) undergoes a simple Hopf bifurcation when $K=-S$ using the criterion presented in \cite{Liu_1994} that involves the properties of the coefficients of characteristic equation instead of those of eigenvalues. In detail, if 
\begin{equation} \label{23}
\begin{array}{l}
(H1)\quad A_{3},A_{4}>0,\;\,A_{2}A_{3}-A_{1}A_{4}>0,\;\,A_{1}A_{2}A_{3}-{A_{3}}^{2}-{%
A_{1}}^{2}A_{4}=0,\medskip \\
(H2)\quad \left. \displaystyle \frac{d}{dK}\left( A_{1}A_{2}A_{3}-{A_{3}}^{2}-{A_{1}}%
^{2}A_{4}\right) \right\vert _{K_{0}}>0,%
\end{array}    
\end{equation}
hold, then there is a simple Hopf bifurcation at $K=K_0$.

Owing to Eq.~(\ref{21}), condition $(H1)$ supplies the bifurcation value $\bar{K}_{HB}=-S$, whereas condition $(H2)$ reduces to $8S^{5}>0$. In particular, when $K=\bar{K}_{HB}$ the eigenvalues (\ref{18}) become
\begin{equation}
  -S \pm \sqrt{(S-1) S},\quad \pm i \sqrt{S}. 
\end{equation}
This proves \textit{(2c)}.

As for \textit{(2d)}, this result has already been proved in Proposition~\ref{prop2} for all $T$. In this case, the two bifurcation values $K_{TB,1}$ and $K_{TB,2}$ reduce to $\bar{K}_{TB}$ and the non-coexistence equilibria exchange the stability properties with $\bar{E}_3$.

\textit{(3)} \;The characteristic polynomial of the Jacobian matrix evaluated at the coexistence equilibria (\ref{13}) writes
 
 \begin{equation} \label{25}
  \lambda^4 +C_1  \lambda^3 +C_2  \lambda^2+C_3 \lambda +C_4,
 \end{equation}
 where $0<B^*_1<S$ and 
 \begin{equation} \label{26}
 \begin{aligned}
  &C_1=2-2 K, \quad C_2={{B^*_1}}^2+{B^*_1} (K-1)+K^2-3 K+2, \\
  &C_3=2 {{B^*_1}}^2+2 {{B^*_1}} (K-1)+(K-1)^2, \quad C_4=0.
 \end{aligned}
 \end{equation}
 Then, since $C_4=0$, at least one of the Routh--Hurwitz conditions is violated and the equilibria (\ref{13}) are locally unstable.
 
 \end{proof}

Finally, we note that the change rates of the biomass and toxicity of the equilibrium $\bar{E}_3$ with respect to the parameters $S$ and $K$ write
\begin{equation}\label{27}
\begin{array}{l}
\dfrac{d\bar{B}_{i,3}^{\ast }}{dS}=\dfrac{K+1}{(K+S+1)^{2}},\quad \dfrac{%
d\bar{T}_{i,3}^{\ast }}{dS}=-\dfrac{1}{(K+S+1)^{2}}, \medskip \\
\dfrac{d\bar{B}_{i,3}^{\ast }}{dK}=-\dfrac{S}{(K+S+1)^{2}},\quad \dfrac{%
d\bar{T}_{i,3}^{\ast }}{dK}=-\dfrac{1}{(K+S+1)^{2}},\quad \quad i=1,2.%
\end{array}
\end{equation}
Then, whenever $\bar{E}_3$ exists and is not unstable,  i.e. at least one of the conditions \textit{(2b)} and  \textit{(2c)} of Proposition \ref{prop3} is verified, we obtain the following results

\begin{equation}\label{28}
\begin{array}{l}
\dfrac{d\bar{B}_{i,3}^{\ast }}{dS}>0\quad \Longleftrightarrow \quad \left(\bar{K}_{HB}\leq
K<\bar{K}_{TB},\;-1<K\leq 0\right) \quad \text{or\quad }0\leq K<\bar{K}_{TB}
,\medskip  \\
\dfrac{d\bar{T}_{i,3}^{\ast }}{dS}<0\quad \Longleftrightarrow \quad \left( \bar{K}_{HB}\leq
K<\bar{K}_{TB},\;K<0\right) \quad \text{or\quad }0\leq K<\bar{K}_{TB} ,%
\end{array}
\end{equation}
and 
\begin{equation}\label{29}
\dfrac{d\bar{B}_{i,3}^{\ast }}{dK}, \; \dfrac{d\bar{T}_{i,3}^{\ast }}{dK} <0\quad
\Longleftrightarrow \quad \left(\bar{K}_{HB}\leq K<\bar{K}_{TB},\;K<0\right) \quad \text{or\quad }%
0\leq K<\bar{K}_{TB}.
\end{equation}
The above results refer to environmental and/or species-specific conditions. In particular, owing to (\ref{9}), the parameter $S$ includes both environmental ($k_1/c_1$) and species-specific conditions that affect the process of growing of biomass positively and the toxicities negatively. On the other hand, the parameter $K$ is only due to species-specific conditions ($s_{ij}/s_{jj}$), consequently its effect is negative for the process of growing of both biomass and toxicities.

\subsection{Linear stability analysis of coexistence equilibria in the case $T\neq1$: numerical analysis and bifurcation diagrams} \label{ssec:4.4}

In Proposition~\ref{prop2} we have been able to prove, for all positive $T$, the occurrence of two transcritical bifurcations at $K_{TB,1}$ and $K_{TB,2}$, where the non-coexistence equilibria exchange the stability properties with the coexistence equilibrium $E_3$ in a neighbourhood of these bifurcation values.  The local character of these stability properties of $E_3$ is confirmed by Proposition \ref{prop3}, where the analysis highlighted the occurrence of an Hopf bifurcation for $T=1$ at $\bar{K}_{HB}<\bar{K}_{TB}$. Then, in this case, the exchange of the stability properties occurs only for  $K \in (\bar{K}_{HB},\bar{K}_{TB})$ and $K>\bar{K}_{TB}$. 

When carrying out the linear stability analysis of the coexistence equilibrium $E_3$ for $T\neq1$, we immediately realize that it is not possible to further extend the analytical results obtained in Proposition~\ref{prop2}. Then, we perform a numerical stability analysis, which also  allows us to locate the Hopf bifurcation values ${K}_{HB}$, for each value of $T$ and $S$.   

In order to detect the stability properties of $E_3$, we identify three relevant regions in $(K,S,T)-$space where $K\in [-4.5,1.5], \; S,T\in[0,2]$ (see Fig.~\ref{fig:stabreg}). In detail, Fig.~\ref{fig:stabreg}(a) shows where $E_3$ is asymptotically stable by solving numerically the Routh--Hurwitz conditions \eqref{20}. Moreover, by using conditions (\ref{23}) and \cite[Theorem 4.1]{Chavez_2004}, we obtain the surfaces where transcritical and Hopf bifurcations, respectively, occur (see Figs.~\ref{fig:stabreg}(b)-(c)).

Panels~(d)--(f) of Fig.~\ref{fig:stabreg} represent 2D sections of $(K,S,T)-$space for $T=0.1,1,1.5$, which combine the information provided in Fig.~\ref{fig:stabreg}~(a)--(c). In particular, in  Fig.~\ref{fig:stabreg}~(d)--(f) the asymptotically stable regions are depicted in gray, whereas the Hopf and  transcritical bifurcation loci are displayed in purple ($K_{HB}$ values) and yellow ($K_{TB,1}$) and brown ($K_{TB,2}$), respectively. We remark that the gray regions are always bounded by the Hopf and one of the transcritical bifurcation curves, except for the case $T=1$, where the two transcritical bifurcation curves coincide. In agreement with ecological expectations, the asymptotically stability regions increase with $T$.
   \begin{figure}[H]
      \centering
      \begin{minipage}{.33\textwidth}
	\centering
	\includegraphics[scale=0.38]{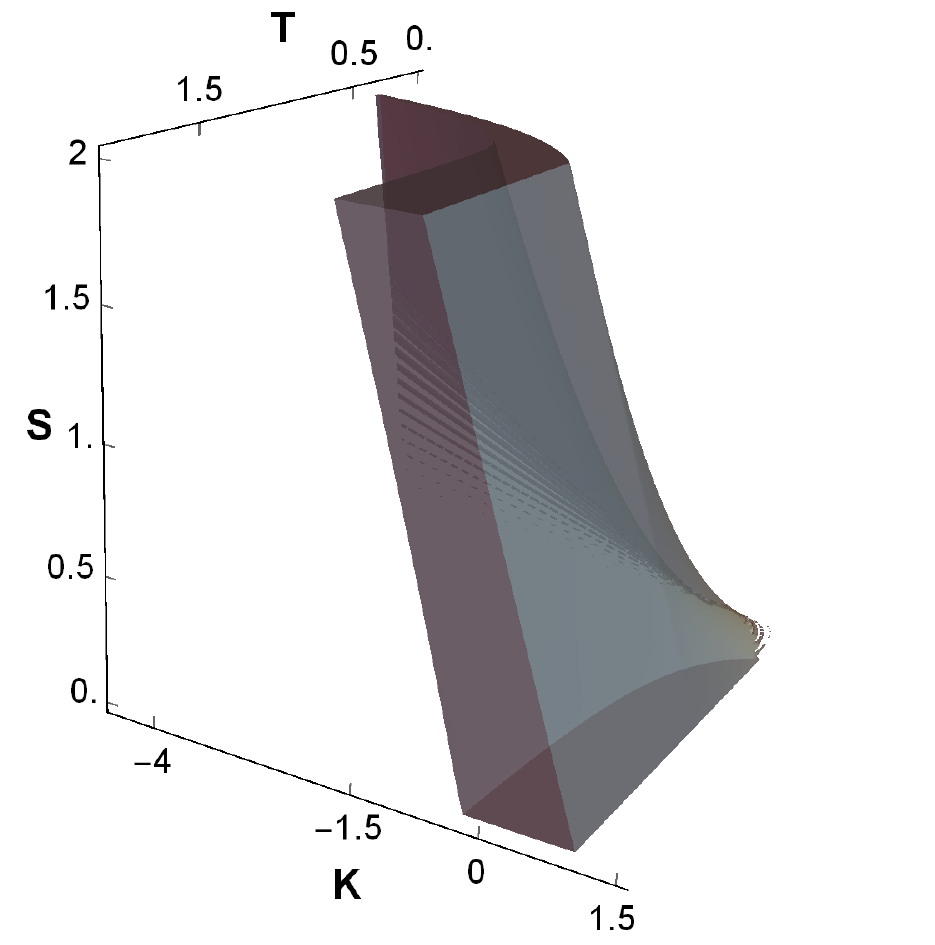} \\
	(a)
      \end{minipage}  
      \hspace{-.6cm}
      \begin{minipage}{.33\textwidth}
	\centering
	\includegraphics[scale=0.38]{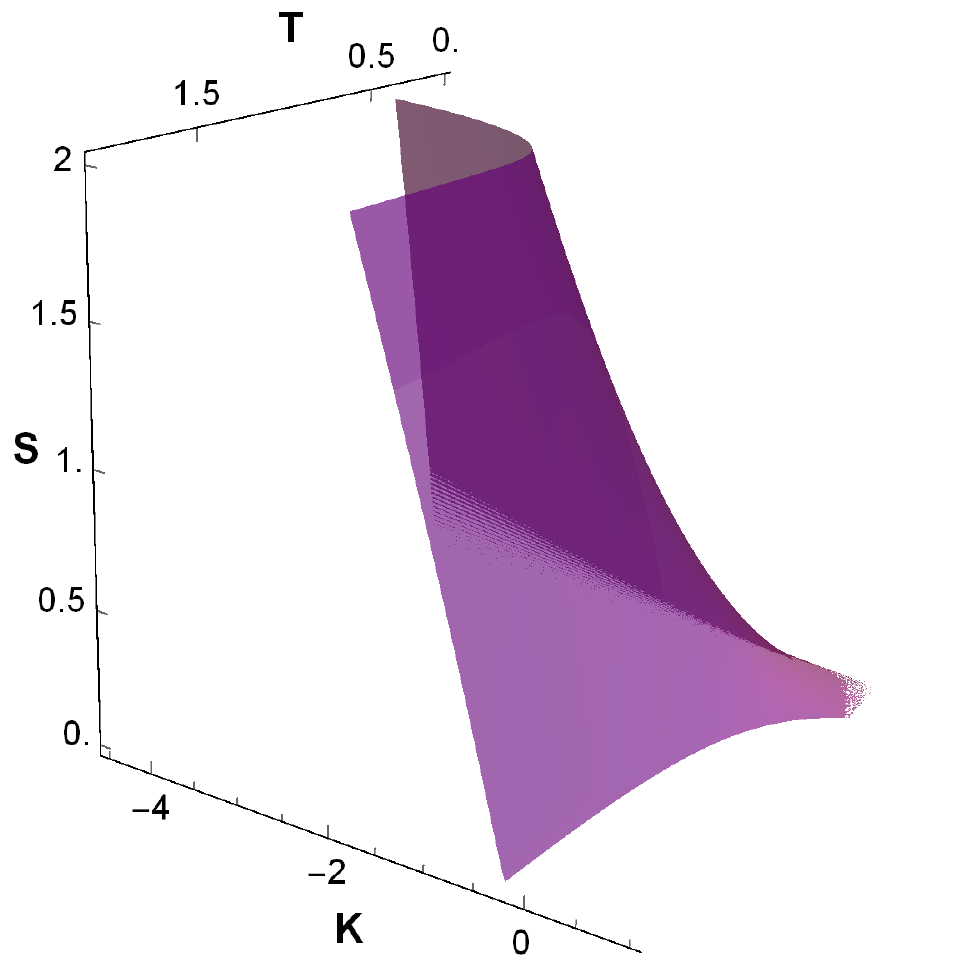}\\
	(b)
      \end{minipage}  
      \hspace{-.6cm}
      \begin{minipage}{.33\textwidth}
	\centering
	\includegraphics[scale=0.38]{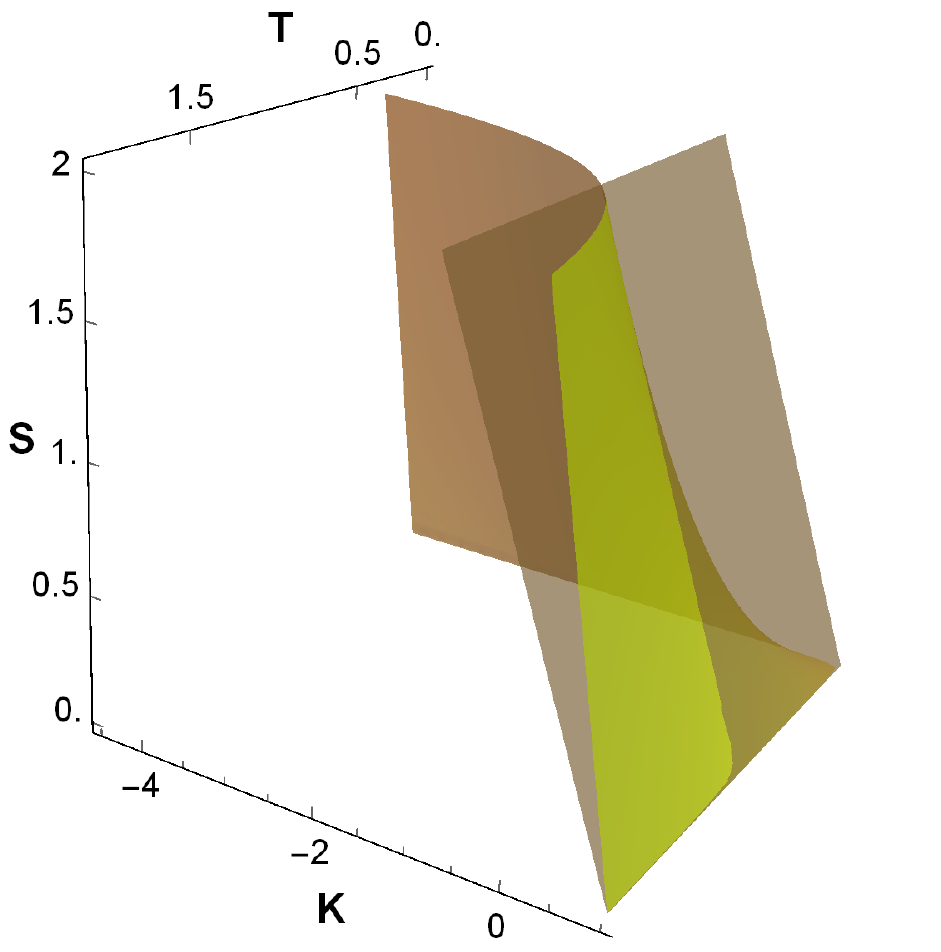} \\
	(c)
      \end{minipage}    
      \\
      \vspace{.5cm}
      \hspace{-.5cm}
      \begin{minipage}{.33\textwidth}
	\centering
	\includegraphics[scale=0.38]{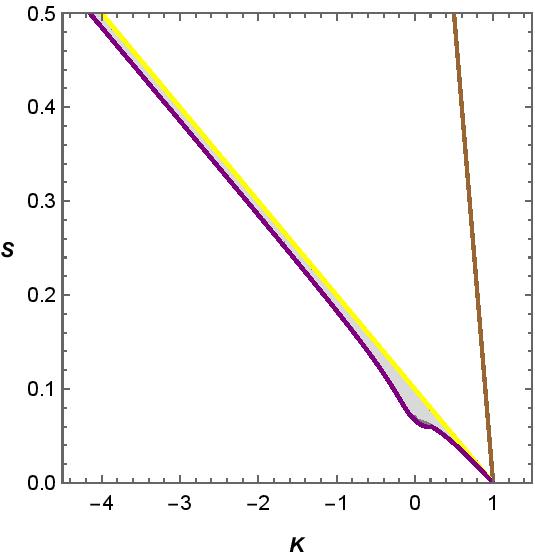}\\
	(d)
      \end{minipage}
      \begin{minipage}{.33\textwidth}
	\centering
	\includegraphics[scale=0.38]{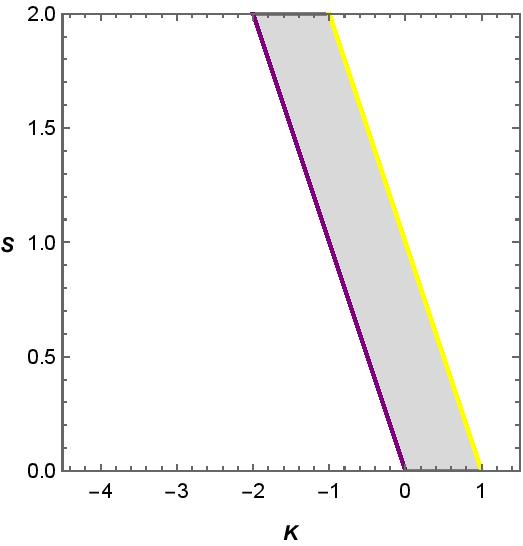}\\
    (e)
      \end{minipage}  
    \begin{minipage}{.33\textwidth}
	\centering
	\includegraphics[scale=0.38]{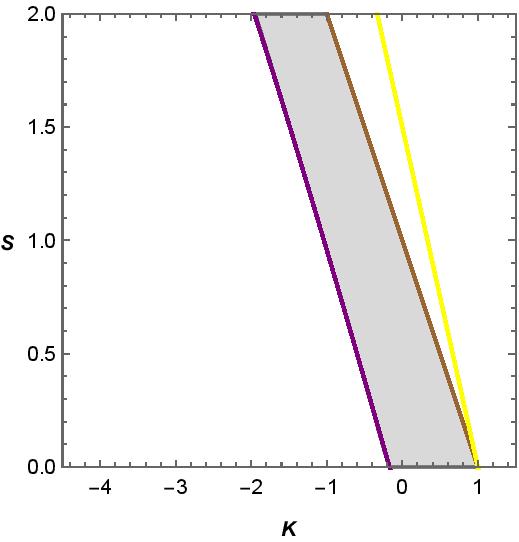} \\
	(f)
      \end{minipage}  
      \caption{(a) Asymptotically stable region for the coexistence equilibrium $E_3$ in $(K,S,T)$-space. (b) Region corresponding to the points in $(K,S,T)$-space where Hopf bifurcations take place. (c) Surfaces of transcritical bifurcation values $K_{TB,1}$ (yellow) and $K_{TB,2}$ (brown) in $(K,S,T)$-space.  The two surfaces intersect for $T=1$. (d)--(f) 2D sections of $(K,S,T)-$space for $T=0.1,1,1.5$. Gray areas represent asymptotically stable regions, whereas purple and yellow/brown curves correspond to Hopf and transcritical bifurcation values, respectively.} 
     \label{fig:stabreg}      
      \end{figure}

Finally, fixing $S=0.5$, we (numerically) construct the bifurcation diagrams into 2D regions  $(K,B_{1}^*)$ and $(K,B_{2}^*)$ for $T=0.1,1,1.5$ (see Fig.~\ref{fig:bifdiag}). In particular, panels in Fig.~\ref{fig:bifdiag} highlight the stability properties of the nontrivial equilibria $E_1, E_2$, and $E_3$ together with the Hopf (triangle) and transcritical (circle)  bifurcation points. Our computations show that
\begin{itemize}
    \item[--] for $T=0.1$: \; \; $K_{HB}\approx-4.15,\; K_{TB,1} = -4,\; K_{TB,2}= 0.5$;
    \item[--] for $T=1$: \; \; $K_{HB} = -0.5,\, K_{TB,1}=K_{TB,2}= 0.5$;
    \item[--] for $T=1.5$: \; \; $K_{HB} \approx -0.6,\, K_{TB,1}=2/3,\, K_{TB,2}=0.5$.
\end{itemize}
      \begin{figure}[H]
      \begin{center}
          $T=0.1$
      \end{center}
      \begin{minipage}{.5\textwidth}
	\centering
	\includegraphics[scale=0.38]{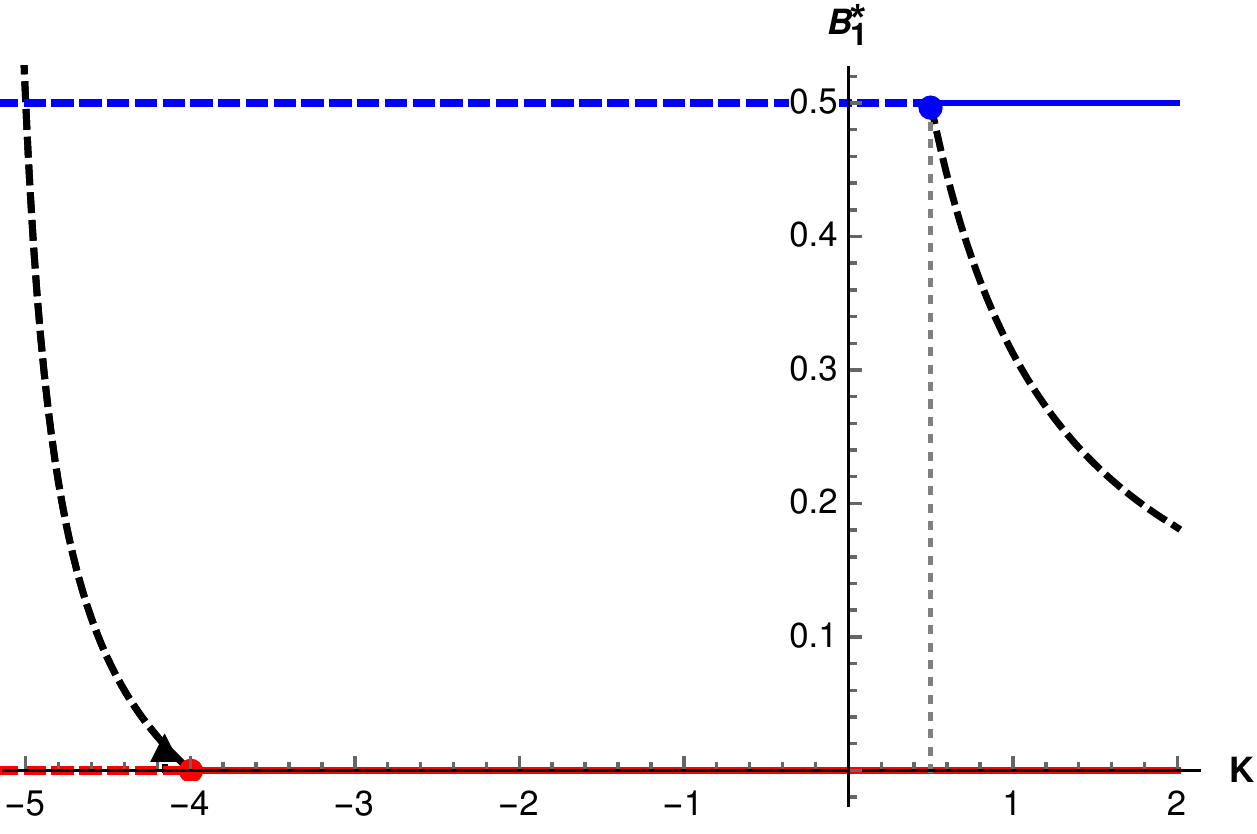}
      \end{minipage}
      \begin{minipage}{.5\textwidth}
	\centering
	\includegraphics[scale=0.38]{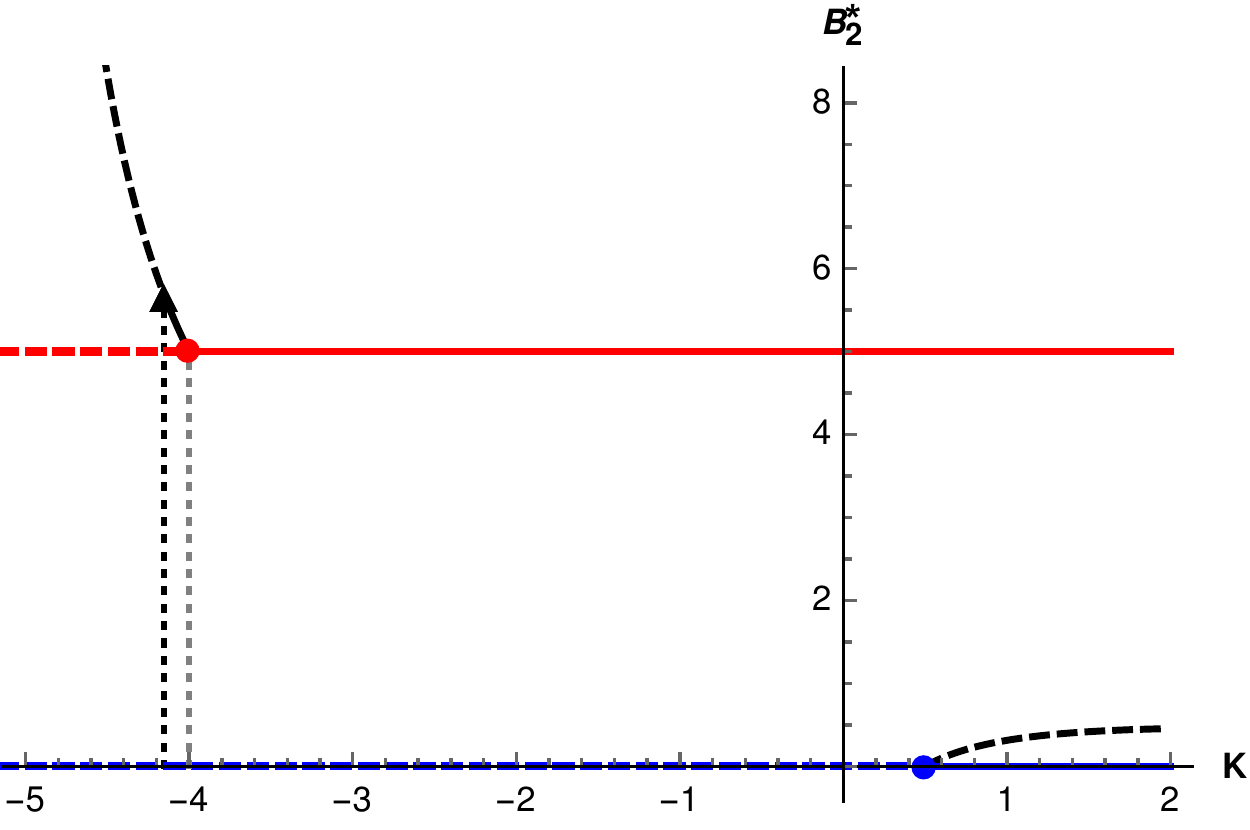}
      \end{minipage}\\
      \begin{center}
          $T=1$
      \end{center}
      \begin{minipage}{0.5\textwidth}
	\centering
	\includegraphics[scale=0.38]{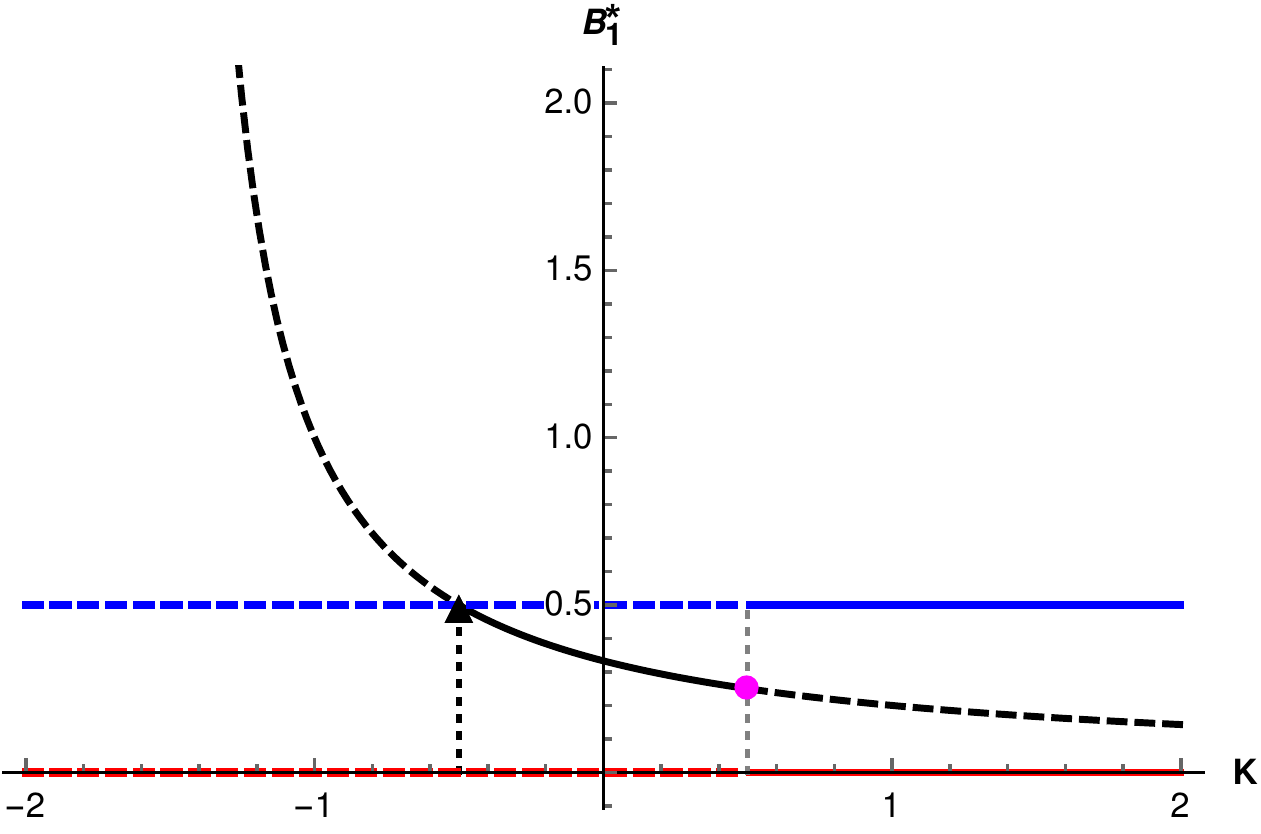}
      \end{minipage}
      \begin{minipage}{0.5\textwidth}
	\centering
	\includegraphics[scale=0.38]{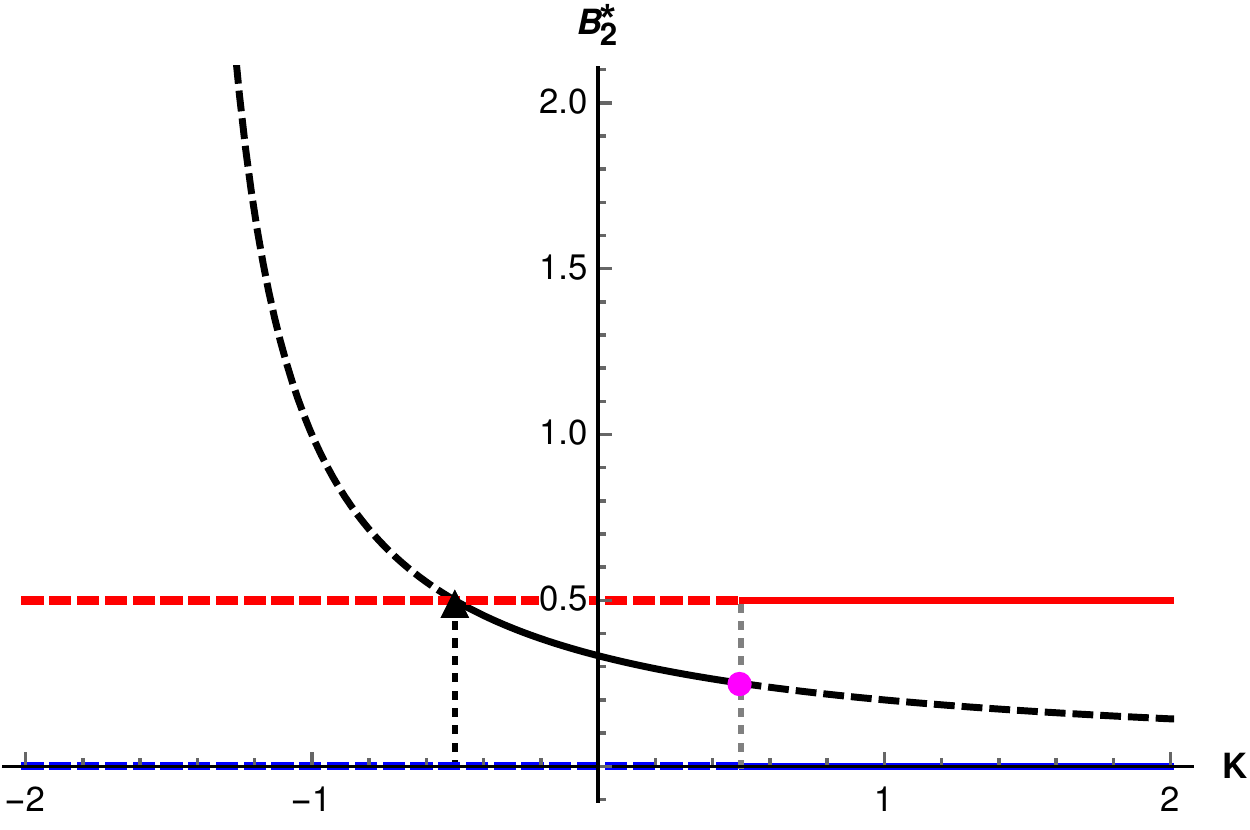}
      \end{minipage}\\
      	\begin{center}
          $T=1.5$
      \end{center}
    \begin{minipage}{0.5\textwidth}
	\centering
	\includegraphics[scale=0.38]{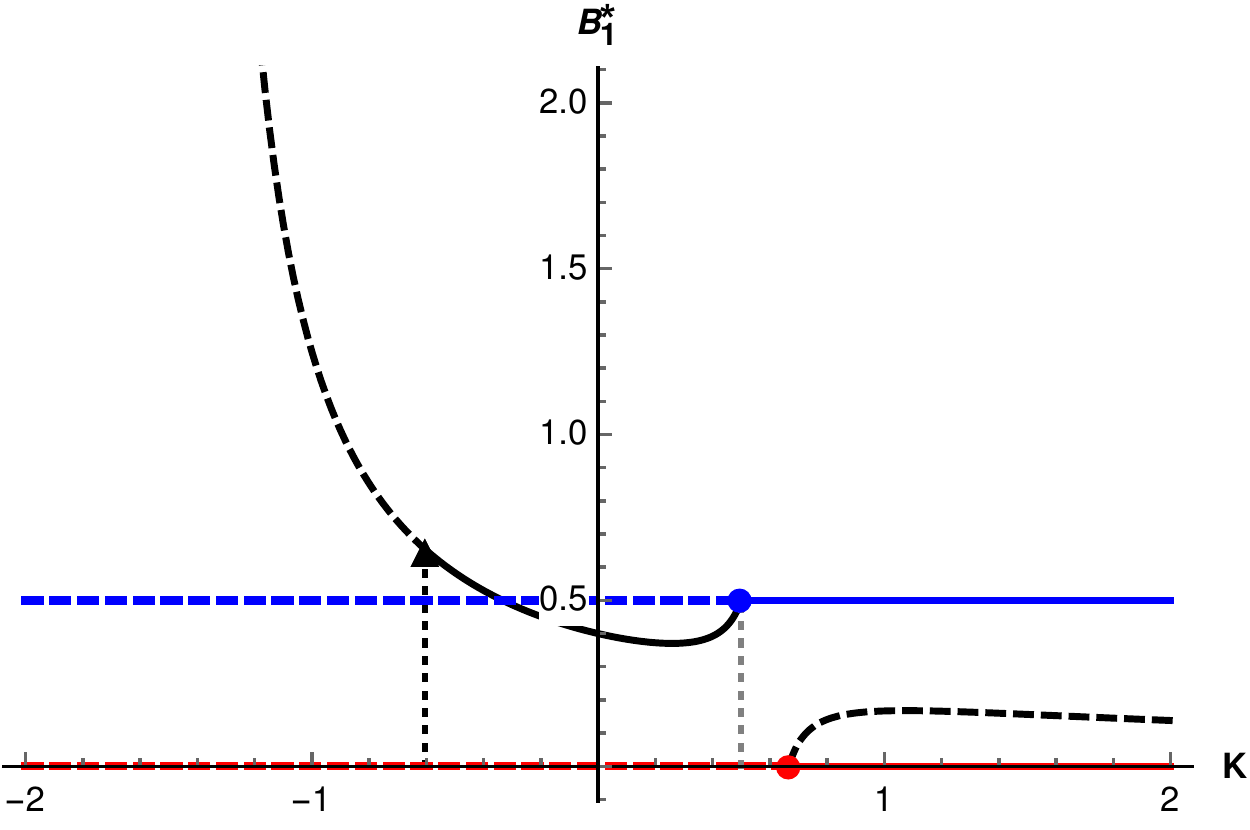}
      \end{minipage}
      \begin{minipage}{.5\textwidth}
	\centering
	\includegraphics[scale=0.38]{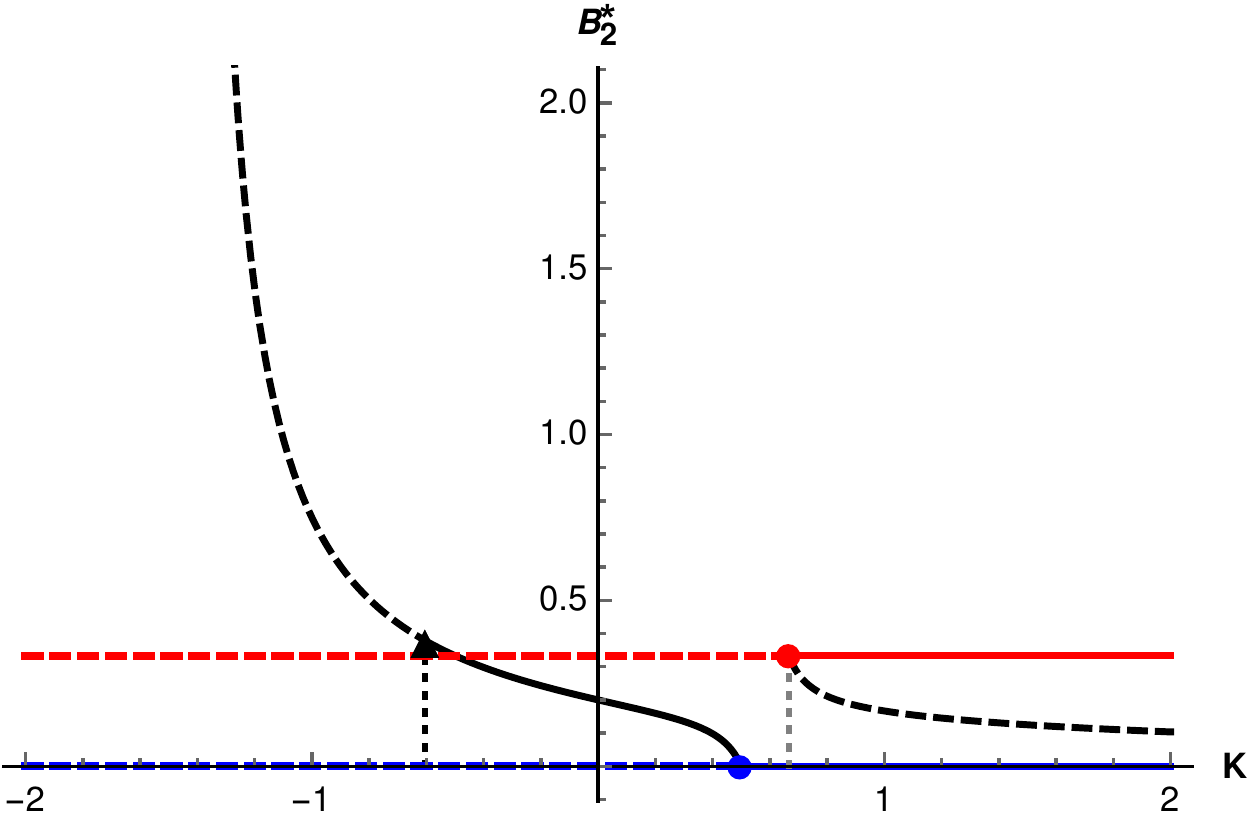}
      \end{minipage}
      \caption{Bifurcation diagrams for the components $B_{1}^*$ and $B_{2}^*$ of the equilibria $E_1$ (red), $E_2$ (blue), and $E_3$ (black) for $S=0.5$. Solid and dashed lines indicate asymptotic stability and instability, respectively. Triangles and circles represent the Hopf and transcritical bifurcation points, respectively. We have that: (i) for $T=0.1$: $K_{HB}\approx-4.15,\, K_{TB,1} = -4,\, K_{TB,2}= 0.5$; (ii) for $T=1$: $K_{HB} = -0.5,\, K_{TB,1}=K_{TB,2}= 0.5$; (iii) for $T=1.5$:  $K_{HB} \approx -0.6,\, K_{TB,1}=2/3,\, K_{TB,2}=0.5$.}
   \label{fig:bifdiag}      
      \end{figure}


\section{Linear stability analysis under spatially heterogeneous perturbation}\label{sec:5}

In terms of system (\ref{8}) the stability properties discussed in Section~\ref{sec:4} refer to (small) spatially homogeneous perturbations of the nontrivial equilibria $E_j,\;j=1,..,3$.\footnote{We exclude from this analysis the equilibrium $E_0$ since it is always locally unstable.} In particular, system (\ref{8}) possesses any periodic solution of system (\ref{8ODE}) as spatially homogeneous periodic solution, including the ones from Hopf bifurcation detected in Proposition~\ref{prop3} and Section~\ref{ssec:4.4}. However, the stability properties of these periodic solutions with respect to (\ref{8}) could be different from that for system (\ref{8ODE}) because of diffusion. 

In this section, we analyse the effect of diffusion on the stability properties of the  equilibria $E_j$, $j=1,\dots,3$ by investigating the linearized jacobian operator relative to system \eqref{8}
\begin{equation}\label{eq:JacPDEh}
L(K)= \left(
\begin{array}{cccc}
 -B_2-T_1-K T_2+1 +\Delta& -B_1 & -B_1 & -B_1 K \\
 -B_2 & -B_1-K T_1-T_2+1+D_B \Delta & -B_2 K & -B_2 \\
 1 & 0 & -S & 0 \\
 0 & 1 & 0 & -S \\
\end{array}
\right),   
\end{equation}
It is known that the eigenvalues of $L(K)$ are given by the eigenvalues of the following matrices
\begin{equation}\label{eq:JacPDE}
L_h(K)= \left(
\begin{array}{cccc}
 -B_2-T_1-K T_2+1 -h& -B_1 & -B_1 & -B_1 K \\
 -B_2 & -B_1-K T_1-T_2+1-D_B h & -B_2 K & -B_2 \\
 1 & 0 & -S & 0 \\
 0 & 1 & 0 & -S \\
\end{array}
\right),   
\end{equation}
where $h=\left(\frac{j \pi}{l_1}\right)^2+\left(\frac{k \pi}{l_2}\right)^2$, $j,k=0,1,2,\dots$, are the eigenvalues of the opposite of the Laplacian operator.\footnote{The eigenvalue problem associated to the Laplacian operator is
\begin{equation}
\begin{aligned}
    \Delta u &= \lambda u \quad \text{ in } \Omega, \\
    \partial_\nu u &= 0 \quad \text{ on } \partial \Omega,
\end{aligned}
\end{equation}
where $\Omega$ is the rectangular domain $\left[ 0, l_1 \right] \times \left[0, l_2 \right]$.
}

In particular, the stability properties of the stationary homogeneous solutions $E_j$ in the presence of heterogeneous perturbations, can be deduced from the signs of the real part of the eigenvalues of $L(K)$ or, equivalently, of the eigenvalues of the $L_h(K)$ for all $h$.

A Hopf bifurcation occurs if $L(K)$ possesses a pair of complex conjugate eigenvalues $\alpha(K) \pm \imath \, \omega(K)$ with
\begin{equation} \label{eq:HC}
  \textit{(HC)} \quad \quad \omega(K_\text{HB})>0, \quad \alpha(K_\text{HB})=0, \quad \alpha'(K)|_{K=K_\text{HB}} \neq 0,
\end{equation}
(see \cite[Theorem II]{Hassard_82}).  Owing to (\ref{eq:JacPDE}), a Hopf bifurcation occurs if there exists a unique $h$ such that conditions~\eqref{eq:HC} are satisfied for $L_h(K)$.  \\
In order to study the occurrence of transcritical bifurcations for system~\eqref{8}, we need to verify that the following three conditions hold (see \cite[Sects. I.6-I.7]{Kielhoefer_2004})
\begin{equation}\label{TC}
\begin{aligned}
\textit{(TC1)} \ \ &0 \text{ is a simple eigenvalue of the matrix } L \text{ evaluated at } (E_j,K_\text{TB}), \\ 
\textit{(TC2)} \ \ &\text{a real eigenvalue } \mu(K) \text{ of } L \text{ crosses the imaginary axis at } (E_j,K_\text{TB}) \text{ with } \\ &\text{ nonvanishing speed}, \\
\textit{(TC3)} \ \ &\text{the slope of the non-trivial bifurcating branch at } K_{TB} \text{ is nonvanishing. }
\end{aligned}
\end{equation}

Also in this case, a transcritical bifurcation occurs if there exists a unique $h$ such that  conditions (\ref{TC}) are verified for $L_h(K)$. 

\subsection{Linear stability analysis under spatially heterogeneous perturbation of non-coexistence equilibrium configurations}

\begin{proposition}\label{prop5}
Let $S$, $T$, and $D_B$ be positive and let $K$ be a real parameter. Then, we have that:
\begin{itemize}
    \item[(i)] $E_1$ is asymptotically stable for system \eqref{8} if and only if $\displaystyle K > 1-\frac{S}{T}$.
    \item[(ii)] $E_2$ is asymptotically stable for system \eqref{8} if and only if $\displaystyle K > 1-S$.
\end{itemize}
\end{proposition}
\begin{proof}
\textit{(i)} The eigenvalues associated to the Jacobian matrix $L_h(K)$ evaluated at $E_1$ are:
\begin{equation} \label{eq:ev1pde}
     -S, \; 1-h-K-\frac{S}{T},\; -\frac{1}{2} \left( D_B h + S \pm \sqrt{-4S+(S-D_B h)^2} \right),
\end{equation}
which are negative for all $h \geq 0$ if and only if $\displaystyle K>1-\frac{S}{T}$.

\textit{(ii)} The eigenvalues associated to the Jacobian matrix $L_h(K)$ evaluated at $E_1$ are:
\begin{equation} \label{eq:ev2pde}
    -S,\; 1-D_B h-K-S,\; -\frac{1}{2} \left( h + S \pm \sqrt{(h-S)^2-4S} \right),
\end{equation}
which are negative for all $h \geq 0$ if and only if $\displaystyle K>1-S$.
\end{proof}
\begin{remark}
The results in Proposition \ref{prop2} are completely confirmed for the corresponding PDEs system \eqref{8}. First, Proposition \ref{prop5} highlights that the spatial contributions do not alter the stability properties of the non-coexistence equilibria $E_1$ and $E_2$. Then, the transcritical bifurcations obtained for system \eqref{8ODE} persist for system \eqref{8}, showing that the inclusion of spatial disturbances does not lead to the occurrence of any additional transcritical bifurcation points (see Figure \ref{fig:bifdiagPDE}). In particular, conditions \textit{(TC1)}--\textit{(TC2)} analytically supply $h=0$ for $K=K_{TB,1},\, K_{TB,2}$, whereas  the software pde2path \cite{pde2path} allows us to numerically verify all conditions, including \textit{(TC3)}. Finally, we are able to analytically confirm the absence of Hopf bifurcations on the branches corresponding to equilibria $E_1$ and $E_2$ - conditions \textit{(HC)} are in fact not satisfied in this case.
\end{remark}

\subsection{Linear stability analysis under spatially heterogeneous perturbation of coexistence equilibrium: numerical analysis and bifurcation
diagram}

In this section, we numerically investigate whether spatial disturbances alter the stability properties of $E_3$ as well as the presence of bifurcation points obtained for the corresponding ODE system (\ref{8ODE}) in Section \ref{sec:4}. Numerical results, provided by the software pde2path \cite{pde2path}, show that both stability properties and the Hopf and transcritical bifurcations  
persist for spatially heterogeneous perturbations for every chosen value of $D_B$. In more detail, we obtain the same bifurcation diagrams in the 2D region\footnote{We indicate with $\Vert \cdot \Vert_2$ the $L^2-$norm on an interval of length $1/2$.} $(K,\Vert B_1 \Vert_2)$ fixing $D_B=0.1,1,2$. In view of the above consideration, in Fig.~\ref{fig:bifdiagPDE} we show a unique bifurcation diagram (representing all considered $D_B$ values) for each $T=0.1,1,1.5$. Differently from the bifurcation diagrams for the ODE system \eqref{8ODE} obtained in Section \ref{ssec:4.4}, we observe here the presence of bifurcation branches corresponding to unstable, non-uniform spatial stationary solutions of \eqref{8} (see green curves in Fig.~\ref{fig:bifdiagPDE}). When $T=1$, this branch coincides with part of the unstable branch of $E_2$ (blue curve).

For $K<K_{HB}$, all nontrivial equilibria are unstable. Therefore, for spatially heterogeneous initial data, we expect no Turing spatio-temporal patterns to occur. This scenario will be confirmed by numerical simulations (see Section~\ref{sec:6}).
      \begin{figure}[H]
      \begin{minipage}{.3\textwidth}
	\centering
	        \hspace{.3cm} \small{$T=0.1$}\\
	\includegraphics[scale=0.12]{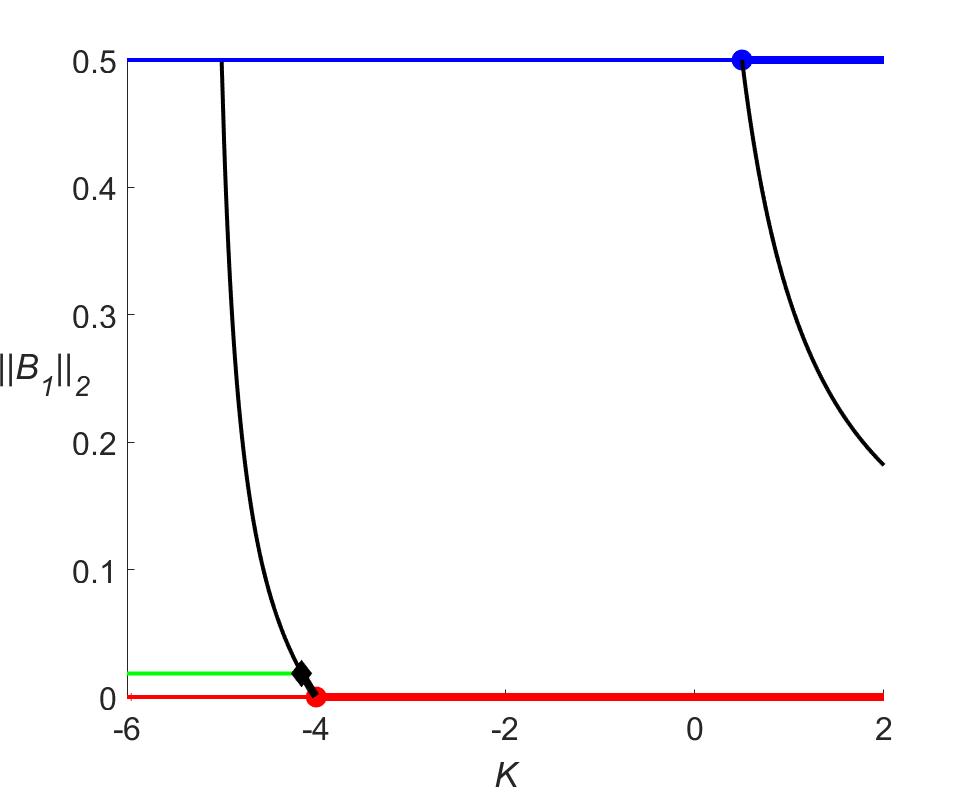}
      \end{minipage}
    \hspace{.3cm}
     \begin{minipage}{.3\textwidth}
	\centering
	          \hspace{.3cm} \small{$T=1$}\\
	\includegraphics[scale=0.12]{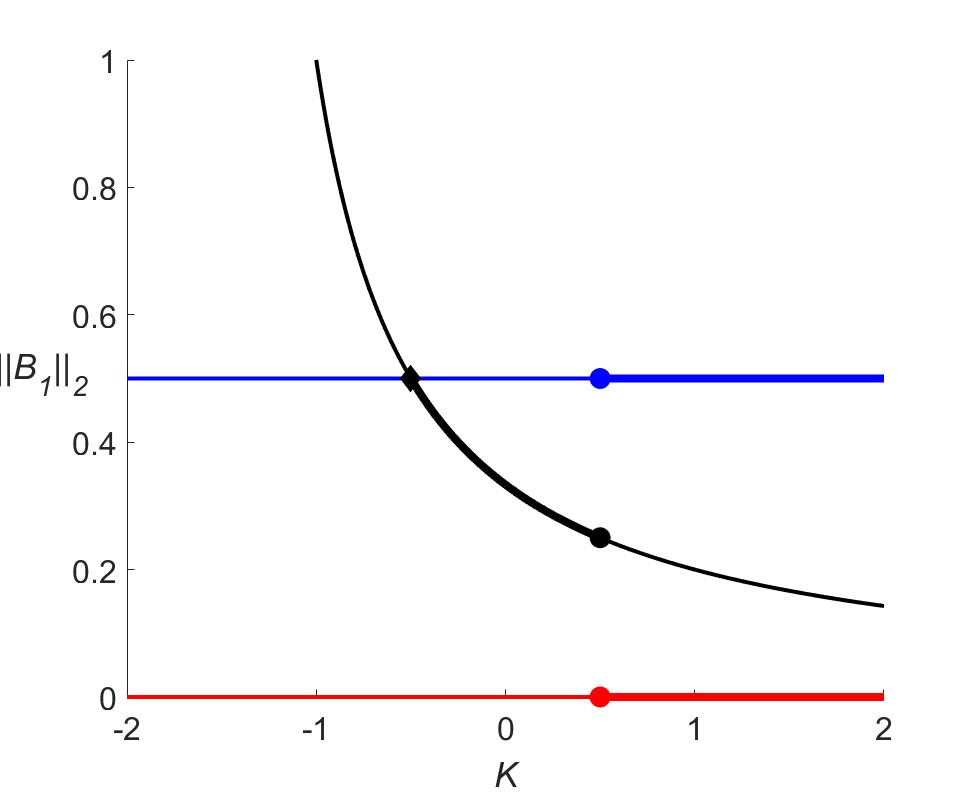}
      \end{minipage}
      \hspace{.3cm}
    \begin{minipage}{.3\textwidth}
	\centering
	          \hspace{.3cm} \small{$T=1.5$}\\
	\includegraphics[scale=0.12]{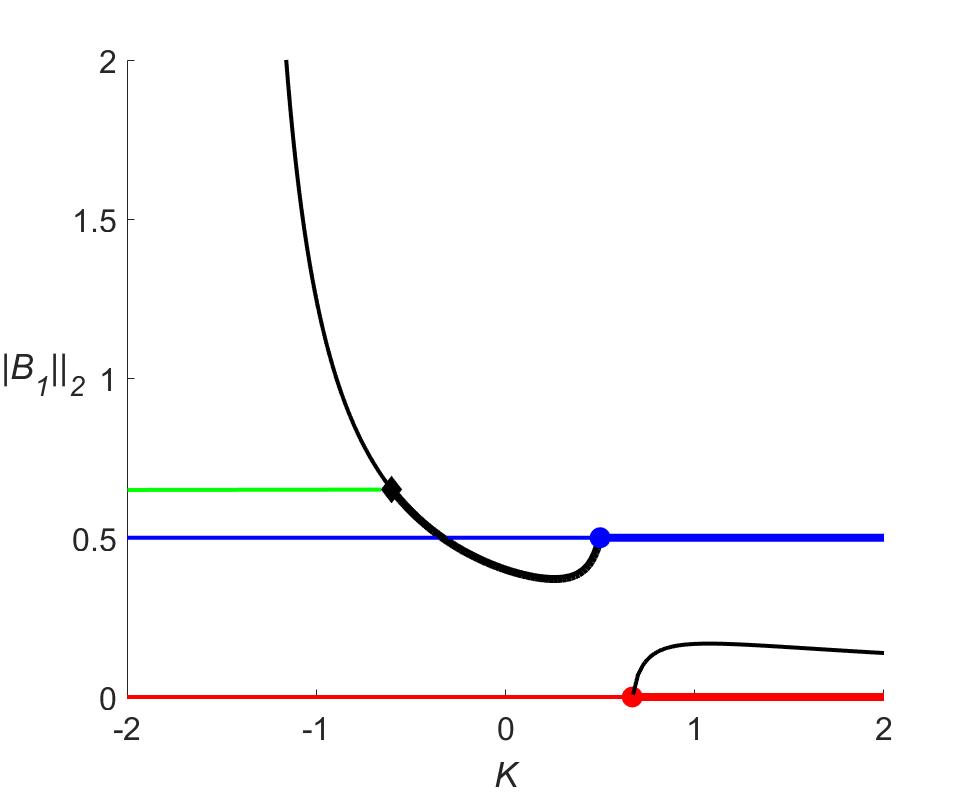}
      \end{minipage}
      \caption{Bifurcation diagrams in the 2D region $(K,\Vert B_1 \Vert_2)$ for system~\eqref{8} obtained with the aid of pde2path~\cite{pde2path} for $T=0.1,1,1.5$ and an arbitrary fixed value of $D_B$. Thick lines indicate asymptotic stability, while thin lines refer to instability in $L^2-$norm of the steady states $E_1$ (red), $E_2$ (blue), and $E_3$ (black). Diamonds indicate Hopf bifurcation points, while circles indicate transcritical bifurcations. The green branches, originating from the Hopf bifurcation points, refer to nonuniform unstable  steady-state solutions.}
   \label{fig:bifdiagPDE}      
      \end{figure}

\section{Pattern formation and numerical simulations} \label{sec:6}

We investigate the dynamics of our model with respect to the its main parameters by performing numerical simulations considering both spatial homogeneity and heterogeneity. In particular, we analyse the effect of $T$ on the evolution of the temporal model \eqref{8ODE} and the influence of $T$ and $D_B$ on the formation of spatio-temporal patterns for system \eqref{8}.

\subsection{ODE simulations} \label{ssec:numODE}

Figure \ref{fig:ODET} shows the effect of $T$ on the evolution of system \eqref{8ODE} in $(B_1,B_2)-$phase space. To this aim, we fix $K=-1$ and $S=0.5$, in order to focus our attention in the proximity of the Hopf bifurcation point for $T=1,\,1.5$. In agreement with the analytical results obtained in Section \ref{ssec:4.3}, when $T=1$ (Fig.~\ref{fig:ODET}, left panel) the two plant species coexist, and converge - after an initial transient phase - to a state where they \textit{symmetrically oscillate}. In fact, as shown in Fig.~\ref{fig:ODET}, the limit cycle (continuous line in the left panel) is symmetric with respect to the diagonal in $(B_1,B_2)-$space. When $T=1.5$, numerical simulations still show the emergence of a \sout{periodic} solution where both species coexist, but the maximum value reached by $B_1$ is higher than the one of $B_2$ (Fig.~\ref{fig:ODET}, right panel). This is due to the fact that species 2 suffers because of its auto--toxicity more than species 1 when $T>1$. We point out that both scenarios occur for negative values of $K$, which correspond to a facilitation in the inter-specific interactions between biomass and toxicity.

The initial condition used in the numerical simulations are $B_1=0.4$ and $B_2=0.6$ and no auto-toxicity for both species at $t=0$. System \eqref{8ODE} was integrated using MATLAB R2012b (the MathWorks) with a variable-order solver (ode15s) based on the numerical differentiation formulas (NDFs), particularly efficient with stiff problems \cite{Shampine_1997}.

\begin{figure}[H]
    \centering
    \includegraphics[scale=0.5]{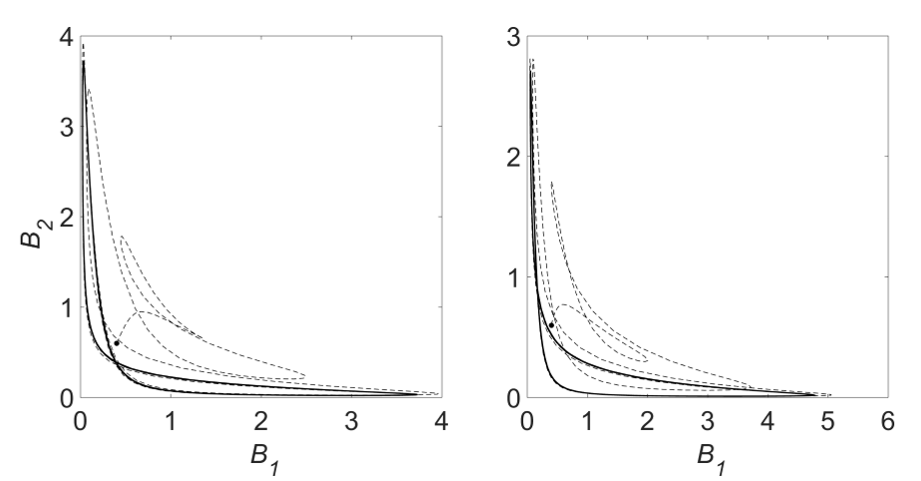}
    \caption{Temporal evolution of system \eqref{8ODE} in $(B_1,B_2)-$phase space. The left panel is obtained for $T=1$, while $T=1.5$ in the right panel. The other parameters values are $K=-1$ and $S=0.5$. At $t=0$, the initial conditions are $B_1=0.4$, $B_2=0.6$, $T_1=T_2=0$. The black dot corresponds to the initial datum, the dashed curve represents the transient phase, and the continuous curve indicated the limit cycle.}
    \label{fig:ODET}
\end{figure}

\subsection{PDE simulations and spatio--temporal patterns} \label{ssec:numPDE}

In this section, we illustrate the numerical results related to the simulation of system~\eqref{8}. We focus our attention on two scenarios which are particularly interesting from the ecological viewpoint: species coexistence and oscillatory behavior. To this aim, we analyze the effect of the parameters $T$ and $D_B$ on the evolution of the system, and fix the value of the other parameters $K=-1$ and $S=0.5$. 
All simulations were performed using zero-flux Neumann boundary conditions on a square lattice of $100 \times 100$ units. The initial conditions considered here for $B_1$ and $B_2$ are two Gaussian distributions with peak $0.25$ in the top-right and bottom-left corner of the domain, respectively, while $T_1$ and $T_2$ are both uniformly zero. All simulations are performed for 500 time steps.
The system~\eqref{8} were integrated using MATLAB R2012b (the MathWorks). The numerical scheme is the following: for the biomass equations we apply a FD scheme forward in time and centered (second order) in the space, whereas the toxicity equations are solved by a first order accurate forward Euler scheme. We check a specific stability criterion, that limits the amplification error in the computed solution ad presented in \cite{Campagna_2018}.
Figure \ref{fig:pannello_DT} shows in each point of the spatial domain the normalized value of $B_1-B_2$ (i.e., the difference between the two biomass densities) at the final time step $t_{\text{max}}=500$ for $T=1,\,1.5$ and $D_B=0.1,\,1,\,2$. As revealed by both Fig.~\ref{fig:pannello_DT} and VIDEO\_1-VIDEO\_6 in the supplementary material, system \eqref{8} exhibits  spatio-temporal patterns of $B_1$ and $B_2$. In particular, when $T=1$ and $D_B=0.1,2$, a scenario occurs where the whole domain is entirely covered by $B_1$ and $B_2$, continuously alternating (see VIDEO\_1 and VIDEO\_3). In all other cases, spatio-temporal patterns in the form of \textit{spirals} are occurring (see VIDEO\_2, and VIDEO\_4-VIDEO\_6).
As $D_B$ increases, the qualitative features of the observed dynamics do not change; however, the system converges to such patterned states with a faster rate.

\begin{figure}
    \centering
    \includegraphics[scale=0.5]{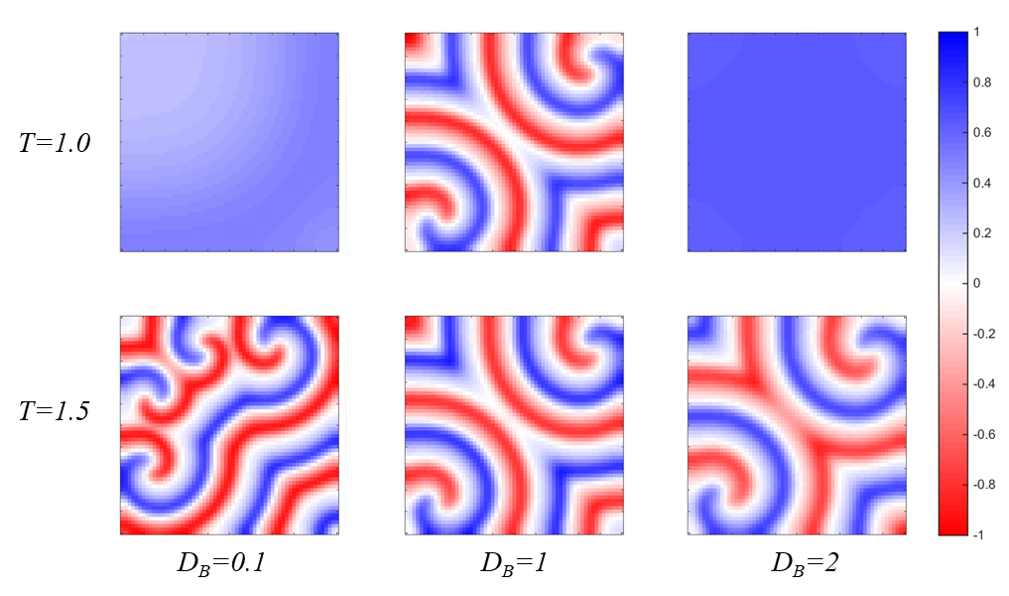}
    \caption{Spatio-temporal patterns of vegetation formed at time $t_{\text{max}} = 500$ for selected values of $T$ (rows) and $D_B$ (columns). The blue-red scale map represents the intensity of the difference between the two biomass densities $B_1 - B_2$, where blue corresponds to the case where $B_1 \gg B_2$ and red indicated $B_2 \gg B_1$). Other parameter values are set to $K=-1$, $S=0.5$. At $t=0$, $B_1$ and $B_2$ correspond to two Gaussian distributions with peak $0.25$ located at the the top-right and bottom-left corner of the domain, respectively. On the other hand, both $T_1$ and $T_2$ are uniformly zero.
}
    \label{fig:pannello_DT}
\end{figure}

In Figure \ref{fig:pannello_tT} we focus our attention on the temporal evolution of the system \eqref{8} starting with the same initial conditions considered in Fig.~\ref{fig:pannello_DT}. Our aim is here to show how the patterns evolve in time. In this case, we limit ourselves to considering $K=-1$, $S=0.5$, and $D_B=0.1$.
Coherently with what stated above, when $T=1$ the system converges to a state where both biomasses $B_1$ and $B_2$ alternate continuously. For $T=1.5$, on the other hand, spiral patterns and waves occur. Both dynamic scenarios persist for longer times, as confirmed by VIDEO\_1 and VIDEO\_4 available in the supplementary material.

\begin{figure}
    \centering
    \includegraphics[scale=0.6]{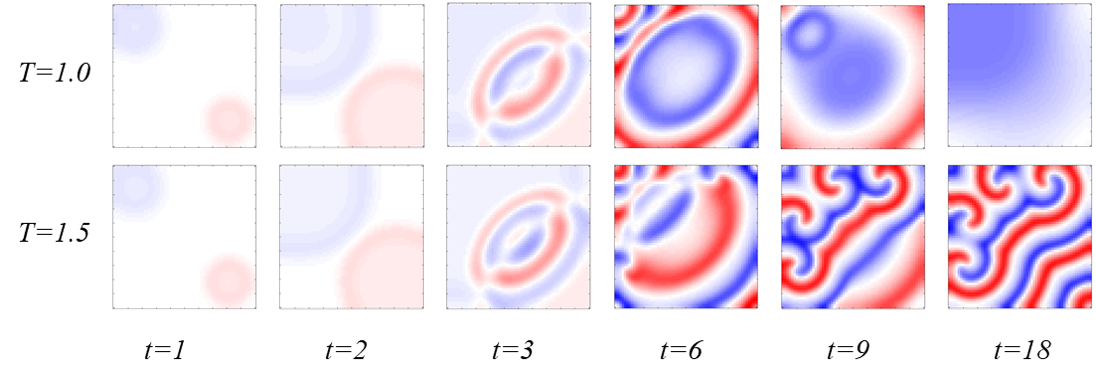}
    \caption{Temporal evolution of system \eqref{8}. The row panels show two-dimensional maps of the normalized difference between the two biomass densities for $T=1,\,1.5$. Other parameter values are $K=-1$, $S=0.5$, and $D_B=0.1$. The initial conditions and the colour scale are the same as in Fig.~\ref{fig:pannello_DT}.}
    \label{fig:pannello_tT}
\end{figure}

\section{Conclusions and research perspectives}

In this paper, we analyzed the effect of both intra- and inter-specific plant-soil feedback induced by toxicity in a spatial model for two plant species. Our theoretical investigation shows the existence of a rich spectrum of ecological scenarios, such as competitive exclusion, stable coexistence, and spatio-temporal patterns. We have analyzed existence and stability of steady-states with respect to all parameters, with a particular focus on the symmetrical scenario $T=1$. A combination of analytical and numerical methods has allowed us to construct bifurcation diagrams both for the temporal model \eqref{8ODE} and for the full spatial model \eqref{8}, as well as to detect Hopf and transcritical bifurcations. The numerical analysis carried out with the aid of the software pde2path has lead to the conclusion that different diffusion coefficients do not alter the structure of the bifurcation diagram. 
Moreover, numerical simulations on a 2D domain reveal the occurrence of spatio-temporal patterns in a domain of influence of the Hopf bifurcation on varying the parameters $T$ and $D_B$.

From an ecological viewpoint, the evolution of biomass and toxicity in the same environment for some parameter ranges induces the formation of spatio-temporal patterns and waves, as already observed in \cite{Carteni_2012,Marasco_2014,Valentin_1999}. In these patterns the species are continuously alternating in space and time without reaching a stable configuration. Moreover, our results reveal that inter-specific plant-soil feedback can lead to stable coexistence of two plant species, whereas the absence of such feedback ($K=0$) would imply competitive exclusion (see, for instance, the case $T=0.1$). In other words, inter-specific plant-soil feedback provides a facilitation mechanism and enlarges the range of species coexistence and then biodiversity. 

In this work, we focused our attention on humid environments, where the intensity of vegetation-water feedbacks is less prominent. In the future, we plan to extend the spectrum of ecological scenarios considered by our model including arid environments where water dynamics must be take into account by means of an additional equation. Finally, an additional  research aim is to generalize our model by considering the interaction among  $n$ species, with $n>2$.


\section*{Appendix A. Theorem 4.1 of \cite{Chavez_2004}}
\label{AppendixA}

Let us consider a general system of ODEs with a parameter $\phi$:
\begin{equation}
\dot{x} = f(x,\phi), \qquad f: \mathbb{R}^n \times \mathbb{R}, \quad f \in C^2(\mathbb{R}^n \times \mathbb{R}). \tag{A1}
\end{equation} \label{ccsys} 
Without loss of generality, we assume that $x = 0$ is an equilibrium for (A1), i.e. $f(0,\phi)=0$ for all $\phi$.

\begin{theorem*}[Theorem 4.1, of \cite{Chavez_2004}] \label{thm:cc}
Assume:
\begin{enumerate}
    \item[(a1)] $A=D_xf(0,0)$ is the linearization matrix of system (A1) around the equilibrium $x=0$ with $\phi$ evaluated at $0$. Zero is a simple eigenvalue of $A$ and all other eigenvalues of $A$ have negative real parts;
    \item[(a2)] Matrix $A$ has a (nonnegative) right eigenvector $\mathbf{w}$ and a left eigenvector $\mathbf{v}$ corresponding to the zero eigenvalue.
\end{enumerate}
Let $f_k$ denote the $k$th component of $f$, and
\begin{equation}\label{A2b}
    a = \sum_{k,i,j=1}^n v_k w_i w_j \frac{\partial^2 f_k}{\partial x_i \partial x_j}(0,0), \qquad b = \sum_{k,i=1}^n v_k w_i \frac{\partial^2 f_k}{\partial x_i \partial \phi}(0,0).\tag{A2}
\end{equation}
Then, the local dynamics of system (A1) around $x=0$ are totally determined by $a$ and $b$.
\begin{enumerate}
    \item[i.] $a>0$, $b>0$. When $\phi < 0$ with $|\phi| \ll 1$, $0$ is locally asymptotically stable, and there exists a positive unstable equilibrium; when $0 < \phi \ll 1$, $0$ is unstable and there exists a negative and locally asymptotically stable equilibrium.
    \item[ii.] $a<0$, $b<0$. When $\phi < 0$ with $|\phi| \ll 1$, $0$ is unstable; when $0 < \phi \ll 1$, $0$ is locally asymptotically stable and there exists a negative unstable equilibrium\footnote{In the version of Theorem 4.1 of~\cite{Chavez_2004} reported here, we have corrected the typo present in the original paper, where the unstable equilibrium is said to be positive.}.
    \item[iii.] $a>0$, $b<0$. When $\phi < 0$ with $|\phi| \ll 1$, $0$ is unstable, and there exists a locally asymptotically stable negative equilibrium; when $0 < \phi \ll 1$, $0$ is stable and a positive unstable equilibrium appears.
    \item[iv.] $a<0$, $b>0$. When $\phi$ changes from negative to positive, $0$ changes its stability from stable to unstable. Correspondingly a negative unstable equilibrium becomes positive and locally asymptotically stable. 
\end{enumerate}
\end{theorem*}

\begin{remark}
As illustrated in Remark 1 in~\cite{Chavez_2004}, we observe that if the equilibrium of interest in the above theorem is a non negative equilibrium $x_0$, then the requirement that $\mathbf{w}$ is non negative is not necessary. When some components in $\mathbf{w}$ are negative, one can still apply the theorem provided that $w_j > 0 $ whenever $(x_0)_j = 0$; instead, if $(x_0)_j>0$, then $w_j$ need not to be positive. Here $w_j$ and $(x_0)_j$ denote the $j$th components of $\mathbf{w}$ and $\mathbf{x}_0$, respectively.
\end{remark}

\begin{acknowledgements}
AI acknowledges financial support from the Austrian Academy of Sciences {\"O}AW via the Multiscale modeling and simulation of crowded transport in the life and social sciences Group NST-0001. The authors would also like to thank Hannes Uecker for the fruitful discussions concerning the implementation of the bifurcation analysis with pde2path.
\end{acknowledgements}

\section*{Conflict of interest}
On behalf of all authors, the corresponding author states that there is no conflict of interest.

\bibliographystyle{spmpsci}
\bibliography{references}

\end{document}